\documentclass[a4paper,twocolumn,11pt,accepted=2021-02-19]{quantumarticle}
\pdfoutput=1
\usepackage{tikz}
\usepackage{pgfplots}
\usepackage{tcolorbox}
\usepackage{tabu}
\usepackage[utf8]{inputenc}
\usepackage{latexsym}
\usepackage{multirow,booktabs}
\colorlet{myPurple}{blue!40!red}
\usepackage{hyperref}
\usepackage{scalerel}

\usepackage{amsfonts,amsmath,amssymb,amssymb,amsthm}
\usepackage{xcolor}
\colorlet{myPurple}{blue!40!red}
\colorlet{myCyan}{cyan!50!gray}
\usepackage[english]{babel}
\usepackage{pgfplots}
\usetikzlibrary{graphs}
\usepackage{graphicx,color,array}
\usepackage{appendix}
\usepackage{verbatim}

\usepackage{changes}

\newcommand{\ba}{\begin{eqnarray}}
\newcommand{\be}{\begin{equation}}
\newcommand{\ee}{\end{equation}}

\newcommand{\bn}{\begin{equation*}}
\newcommand{\en}{\end{equation*}}

\newcommand{\ea}{\end{eqnarray}}
\newcommand{\ban}{\begin{eqnarray*}}
\newcommand{\ean}{\end{eqnarray*}}
\newcommand{\Tr}{\operatorname{tr}}
\newcommand{\tr}{\operatorname{tr}}

\newcommand{\Si}[1]{\mathcal{S}_{#1}}

\newcommand{\sket}[1]{{\ensuremath{\lvert#1\rangle}}}
\newcommand{\lket}[1]{{\ensuremath{\left\lvert#1\right\rangle}}}
\newcommand{\ket}[1]{\if@display\lket{#1}\else\sket{#1}\fi}

\newcommand{\sbra}[1]{{\ensuremath{\langle#1\rvert}}}
\newcommand{\lbra}[1]{{\ensuremath{\left\langle#1\right\rvert}}}
\newcommand{\bra}[1]{\if@display\lbra{#1}\else\sbra{#1}\fi}

\newcommand{\sbraket}[2]{{\ensuremath{\langle#1\rvert#2\rangle}}}
\newcommand{\lbraket}[2]{{\ensuremath{\left\langle#1\!\left\rvert\vphantom{#1}#2\right.\!\right\rangle}}}
\newcommand{\braket}[2]{\if@display\lbraket{#1}{#2}\else\sbraket{#1}{#2}\fi}

\newcommand{\sketbra}[2]{{\ensuremath{\lvert #1\rangle\!\langle #2\rvert}}}
\newcommand{\lketbra}[2]{{\ensuremath{\left\lvert #1\right\rangle\!\!\left\langle #2\right\rvert}}}
\newcommand{\ketbra}[2]{\if@display\lketbra{#1}{#2}\else\sketbra{#1}{#2}\fi}

\newcommand{\proj}[1]{\ketbra{#1}{#1}}

\newcommand{\tp}{\otimes}
\newcommand{\idd}{\openone}

\newcommand{\ie}{{\it{i.e.}~}}

\newcommand{\tx}[1]{\textnormal{#1}}

\newcommand{\rA}{\mathrm{A}}
\newcommand{\ra}{\mathrm{a}}

\newcommand{\rB}{\mathrm{B}}

\newcommand{\rP}{\mathrm{P}}

\newcommand{\ab}{\rA\rB\rP}
\newcommand{\abp}{\rA'\rB'}

\newcommand{\M}{\mathsf{M}}
\newcommand{\N}{\mathsf{N}}
\newcommand{\X}{\mathsf{X}}
\newcommand{\Y}{\mathsf{Y}}
\newcommand{\Z}{\mathsf{Z}}
\newcommand{\D}{\mathsf{D}}
\newcommand{\E}{\mathsf{E}}
\newcommand{\K}{\mathsf{K}}
\newcommand{\Ll}{\mathsf{L}}

\newcommand{\Hh}{\mathsf{H}}
\newcommand{\sz}{\sigma_{\mathsf{Z}}}
\newcommand{\sx}{\sigma_{\mathsf{X}}}
\newcommand{\sy}{\sigma_{\mathsf{Y}}}
\newcommand{\aA}{\mathbf{a}}
\newcommand{\bB}{\mathbf{b}}
\newcommand{\xX}{\mathbf{x}}
\newcommand{\yY}{\mathbf{y}}
\newcommand{\rh}{\rho_{\rA\rB}}

\newcommand{\ivan}[1]{{\color{blue} #1}}

\pgfplotsset{compat=1.14}

\newtheorem{theorem}{Theorem}
\newtheorem{lemma}[theorem]{Lemma}
\newtheorem{cor}{Corollary}

\begin{document}


\title{Device-independent certification of tensor products of quantum states using single-copy self-testing protocols}

\author{Ivan \v{S}upi{\'c}}
\affiliation{D{\'{e}}partement de Physique Appliqu\'{e}e, Universit\'{e} de Gen\`{e}ve, 1211 Gen\`{e}ve, Switzerland}
\author{Daniel Cavalcanti}
\affiliation{ICFO-Institut de Ciencies Fotoniques, The Barcelona Institute of Science and Technology, 08860 Castelldefels (Barcelona), Spain}
\author{Joseph Bowles}
\affiliation{ICFO-Institut de Ciencies Fotoniques, The Barcelona Institute of Science and Technology, 08860 Castelldefels (Barcelona), Spain}

\date{05.03.2021.}

\begin{abstract}
Self-testing protocols are methods to determine the presence of shared entangled states in a device independent scenario, where no assumptions on the measurements involved in the protocol are made. A particular type of self-testing protocol, called parallel self-testing, can certify the presence of copies of a state, however such protocols typically suffer from the problem of requiring a number of measurements that increases with respect to the number of copies one aims to certify. Here we propose a procedure to transform single-copy self-testing protocols into a procedure that certifies the tensor product of an arbitrary number of (not necessarily equal) quantum states, without increasing the number of parties or measurement choices. Moreover, we prove that self-testing protocols that certify a state and rank-one measurements can always be parallelized to certify many copies of the state.  Our results suggest a method to achieve device-independent unbounded randomness expansion with high-dimensional quantum states.
\end{abstract}

\maketitle
\emph{Introduction}
--- Bell nonlocality describes measurement correlations which are rigidly incompatible with the notion of local determinism \cite{bell,bellreview}. Namely, all local deterministic theories satisfy bounds---called Bell inequalities---which limit the strength of the correlations between measurement outcomes of two spatially distant and non-communicating systems. Interestingly, it is possible to violate such Bell inequalities in quantum  experiments \cite{loopholefree1,loopholefree2,loopholefree3}. Such violating correlations, called nonlocal, are closely related to quantum resources such as entanglement and measurement incompatibility,
essential for the development of modern day quantum technologies. 

Bell nonlocality also plays a crucial role in so-called device-independent protocols. It turns out that the violation of a Bell inequality 
is a function of the observed correlations alone, regardless of the underlying physical realization. Thus, the sole observation of a Bell inequality violation witnesses the presence of both entanglement and incompatibility without having knowledge or making any assumptions about the underlying experimental implementation. Such an assumption-free verification is named \textit{device-independent} and has a special significance in cryptographic scenarios \cite{acin2007device,pironio2010random}.


The maximal violation of some Bell inequalities can even imply the precise form of the underlying state and measurements. This can be seen as a device-independent certification, and has received the name of self-testing \cite{STreview,Mayers2004}. The simplest example of such phenomenon is the maximal violation of the Clauser-Horn-Shimony-Holt (Bell) inequality \cite{chsh}, which can be used to self-test the maximally entangled pair of qubits and mutually unbiased local measurements \cite{Popescu1992,Tsirelson1993,Summers1987}.

There exist nonlocal correlations that self-test several copies of a quantum state, a process called parallel self-testing. These self-testing protocols have immediate applications in situations where high amounts of entanglement is needed, such as randomness expansion \cite{Coudron:2014:IRE:2591796.2591873}, parallel quantum key distribution \cite{parallelQKD,parallelQKD2}, delegated quantum computing \cite{ruv,leash}, and universal entanglement certification \cite{Bowles2018a}. One drawback of the first parallel self-testing protocols is that they require a number of local measurements that increases exponentially with the number of copies one wants to certify \cite{Wu2016,Coladangelo2017a,Coudron2016,McKague2017a,BSCA}. This fact increases the time-cost and the randomness consumption of the protocol (relevant for cryptographic applications). More recently, techniques to reduce the number of local measurements to poly(log(n)) and log(n) were found \cite{Natarajan,Natarajan2018,Chao2016,Ostrev2016,BSCA}.
There exist protocols to self-test entangled states of arbitrary dimension with a constant number of three or four inputs per party \cite{Yang,Coladangelo2017b}. These protocols, when applied to copies of quantum states require making joint measurements between the local subsystems of each copy, making them challenging from an experimental perspective.

In this manuscript, we show a procedure to combine different self-testing protocols into a protocol that self tests tensor products of quantum states, without increasing the number of required measurements. The combined protocol has the advantage of not requiring joint measurements among the copies. As a key application, we show a way of self-testing $n$ copies of the two-qubit maximally entangled state using only two measurements per party (see Figure \ref{ParST} c). This is the first self-testing protocol using the minimum number of local measurements possible for certifying an unbounded amount of entanglement. This procedure can therefore be used to convert one random bit into an arbitrary number of private random bits.

\begin{figure*}
    \centering
    \includegraphics[width=1.6\columnwidth]{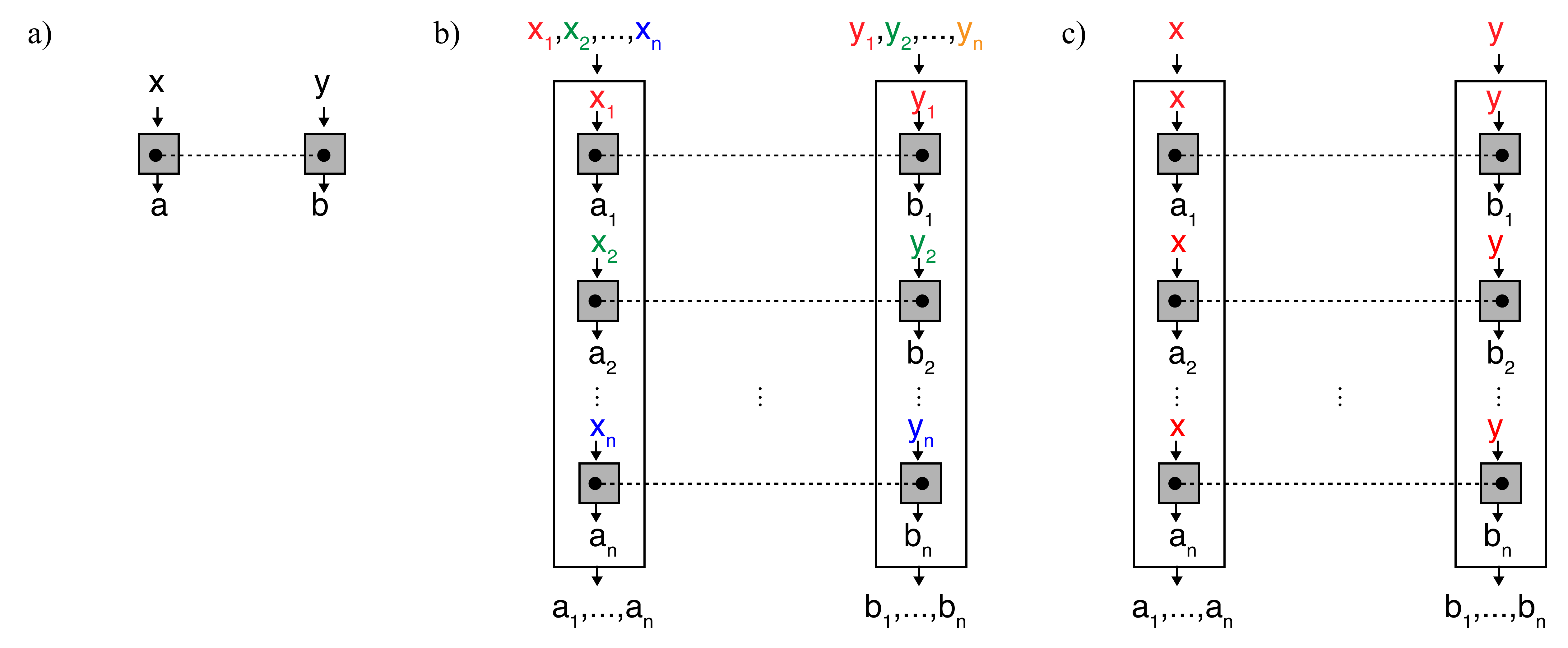}
    \caption{a) Standard Bell scenario, where local measurements labeled by $x$ and $y$ with outcomes $a$ and $b$ are performed on a shared state. b) Standard parallel self-testing protocol certifying $n$ copies of a quantum states. The measurement choice corresponds to a string of $n$ elements, each of them used to determine the measurement applied to each copy.  Thus, the total number of local measurements increases exponentially with the number of copies one aims at certifying. c) The procedure we propose here (see Theorem 1) transforms a self-testing protocol for a single quantum state into a protocol for an arbitrary number of copies. Each party applies the same measurement to each copy and outputs each individual outcome.}
    \label{ParST}
\end{figure*}

\emph{Main idea of the method}-- For simplicity, we explain how our method can transform the self-testing based on the Clauser-Horn-Shimony-Holt (CHSH) game into a self-test for $n$ copies of a two-qubit maximally entangled state. The scenario consists of two space-like separated parties, Alice and Bob, making local measurements on a shared quantum system. Alice and Bob apply one of two measurements each, labeled  $x=0,1$ and $y=0,1$ respectively, and their goal is to obtain outputs $a=0,1$ and $b=0,1$ such that $a \oplus b = x\cdot y$ where $\oplus$ is addition modulo 2. Their score is defined as $\omega=\frac{1}{4}\sum_{x,y}P(a\oplus b=x\cdot y|x,y)$, \ie the probability of satisfying the winning condition averaged over a uniform choice of $x,y$. Alice and Bob can use quantum resources to increase their score. In that case the inputs denote the measurement choices, while the outputs are simply measurement results. When they share quantum state $\rho_{\rA\rB}$, Alice performs measurements $\{\M_{a|x}\}_{a,x}$, and Bob $\{\N_{b|y}\}_{b,y}$ the probabilities $p(a,b|x,y)$ satisfy the Born rule $p(a,b|x,y) = \tr\left[\M_{a|x}\otimes\N_{b|y}\rho_{\rA\rB}\right]$. The measurement observables are defined as $\rA_x = \M_{0|x}-\M_{1|x}$ and $\rB_y = \N_{0|y}-\N_{1|y}$. If Alice and Bob make use of classical resources (or separable states), the best score they can achieve is $\omega = \frac{3}{4}$, which is equivalent to satisfying the CHSH Bell inequality
\begin{align}
    \langle\rA_0\rB_0 \rangle+\langle\rA_0\rB_1 \rangle+\langle\rA_1\rB_0 \rangle-\langle\rA_1\rB_1 \rangle\leq 2,
\end{align}
where $ \langle\rA_x\rB_y \rangle=\sum_{ab}(-1)^{a+b}P(a,b\vert x, y)$. 
On the other hand, if they share a maximally entangled pair of qubits $\ket{\phi^+}=\frac{1}{\sqrt{2}}[\ket{00}+\ket{11}]$, Alice's measurements correspond to the Pauli observables $\rA_0 = \sigma_{\mathsf{Z}}$ and $\rA_1 = \sigma_{\mathsf{X}}$ and Bob's measurements to $\rB_0 = (\sz + \sx)/\sqrt{2} \equiv \sigma_+$ and $\rB_1 = (\sz - \sx)/\sqrt{2}\equiv \sigma_-$, they can achieve the score $\omega_q  \approx 0.8536$. This strategy  violates the CHSH Bell inequality to a value of $2\sqrt{2}$, and is the largest violation possible in quantum theory. This maximum value is also known to self-test the state $\ket{\phi^+}$. Thus, up to a possible local change of basis and extra unused degrees of freedom, the state $\ket{\phi^+}$ is the only state that achieves this value. 

The CHSH inequality can also be used to self-test $n$ copies of $\ket{\phi^+}^{\otimes n}$ via parallel self-testing \cite{Coladangelo2017a,McKague2017a}. For that purpose Alice and Bob in each round receive $n$-bit inputs $\mathbf{x} = (x_1,x_2,\cdots,x_n)$ and $\mathbf{y} = (y_1,y_2,\cdots,y_n)$, and return $n$-bit outputs  $\mathbf{a} = (a_1,a_2,\cdots,a_n)$ and $\mathbf{b} = (b_1,b_2,\cdots,b_n)$. To self-test the state, Alice measures $\sz$ or $\sx$ on her local qubit $i$ depending on $x_i=0,1$, and Bob does similarly, measuring $\sigma_+$ or $\sigma_-$ depending on $y_i=0,1$ and returning outcome $b_i$ (see Fig.\ \ref{ParST}b). The resulting correlations maximally violate $n$ CHSH inequalities in parallel and imply Alice and Bob share a tensor product of $n$ maximally entangled pairs of qubits. 

We consider, in turn, a scenario in which each party has only two choices of measurements $x=0,1$ and $y=0,1$ (Fig.\ \ref{ParST}c). If $x=0$ Alice measures all of her qubits in the $\sz$-basis, whereas if $x=1$ she measures all of them in the $\sx$-basis. In both cases she returns an output consisting of $n$ bits $\mathbf{a} = (a_1,a_2,\cdots,a_n)$ corresponding to the outcomes of each of the $n$ measurements. Bob proceeds similarly, measuring all his qubits in $\sigma_+$ and $\sigma_-$ bases and returning $\mathbf{b} = (b_1,b_2,\cdots,b_n)$. For any pair $(a_i,b_i)$, the condition $P(a_i\oplus b_i=x\cdot y) = \omega_q$ is satisfied, and one might naively think that this information alone is enough to conclude that the state is $\ket{\phi^+}^{\otimes n}$. This would be a mistake however, as the following counter-example shows. Suppose Alice and Bob measure a single copy of $\ket{\phi^+}$, obtaining outputs $a$, $b$ such that $P(a\oplus b=x\cdot y) = \omega_q$. Then, they set $a_i=a$ and $b_i=b$ for all $i$, leading to  $P(a_i\oplus b_i=x\cdot y) = \omega_q$ for all $i$. Thus, one can achieve $P(a_i\oplus b_i=x\cdot y) = \omega_q$ for all $i$ with only a single copy of $\ket{\phi^+}$. 

One thus needs to consider more information than just the CHSH score of each copy. One possibility is to consider the marginal statistics, since in the parallel $n$-copy strategy one has $p(\mathbf{a}\vert x)=1/2^n$ whereas in the single-copy strategy the local output bits $a_i$ are always perfectly correlated (with a similar situation for Bob). However this will not work, since by using $n$ bits of pre-shared randomness $(\lambda_1,\cdots,\lambda_n$) and post-processing their outputs according to $a_i=a+\lambda_i$ and $b_i= b+\lambda_i$, the local output bits of the single-copy strategy can be decorrelated without affecting the CHSH scores. Notice however, that there is still a crucial difference between the two strategies that remains: in each round of the single-copy strategy, if the first pair satisfies $a_1+b_1=x\cdot y$, then all other pairs also satisfy $a_i+b_i=x\cdot y$. This is not the case for the $n$-copy strategy, where each pair has a probability $\omega_q$ to satisfy the condition in each round, independent of the other pairs. It is precisely this difference that we use as an inspiration to design our self-testing protocol.  

\emph{The main result}-- Before stating our main result, let us define self-testing in a precise way. Consider a general bipartite Bell scenario, where $x,y \in \{0,1,\cdots , m-1\}$  and ${a},{b} \in \{0,1, \cdots, o-1\}$.  Alice and Bob share the state $\rho_{\rA\rB}=\tr_\rP\left[\ket{\psi}
\bra{\psi}_{\rA\rB \rP}\right]$, with $\ket{\psi}_{\rA\rB \rP}$ being any purification of $\rho_{\rA\rB}$.

The probabilities of obtaining outputs $a$ and $b$, when the inputs are $x$ and $y$, respectively, are given by
\begin{align}
    p(a,b \vert x,y) = \bra{\psi}\M_{{a}|x}\tp \N_{{b}|y} \tp \openone_\rP  \ket{\psi},
\end{align}
where $\{\M_{{a}|x}\}$ and $\{\N_{{b}|y}\}$ are Alice's and Bob's local measurement operators. We say that the probabilities $\{p(a,b \vert x,y)\}$ self-test the reference experiment $\{\ket{\psi'}_{\rA'\rB'},\M'_{{a}|x},\N'_{{b}|y}\}$ if observing them in an experiment implies the existence of a local unitary transformation $U=U_{\rA\rA'}\tp V_{\rB\rB'}\tp\openone_\rP$ such that 
\begin{align}\label{eq:oneselftest} \nonumber
    &U\left[\ket{\psi}_{\rA\rB\rP}\tp\ket{00}_{\rA'\rB'}\right]  = \ket{\xi}_{\rA\rB\rP}\tp\ket{\psi'}_{\rA'\rB'}; \\ 
    &U\left[{\M}_{{a}|x}\tp{\N}_{{b}|y}\tp\openone_\rP\ket{\psi}_{\rA\rB\rP}\tp\ket{00}_{\rA'\rB'}\right]  = \\ &\qquad \qquad \qquad = \ket{\xi}_{\rA\rB\rP}\tp\left({\M}'_{{a}|x}\tp{\N'}_{{b}|y}\ket{\psi'}_{\rA'\rB'}\right)\nonumber.
\end{align}
These equations state that, up to ancillary degrees of freedom and local basis transformations, the state $\ket{\psi}$ is equivalent to $\ket{\psi'}$, and the measurements $\M_{{a}|x}$ and $\N_{{b}|y}$ act on $\ket{\psi}$ in the same way as $\M'_{{a}|x}$ and $\N'_{{b}|y}$ act on $\ket{\psi'}$. The state $\ket{\xi}$ is usually called the junk state and it can be seen as the state of additional degrees of freedom. Finally, let us define a Bell expression $\mathcal{I}$, as a linear combination of experimental probabilities $\mathcal{I} = \sum_{abxy}b_{a,b}^{x,y}p(a,b|x,y)$, where $b_{a,b}^{x,y} \in \mathbb{R}$. 

Now we can state our main result:

\begin{theorem}\label{theorem1}
Consider a Bell expression $\mathcal{I}$ and $\{p(a,b\vert x,y)\}$ for a scenario where Alice and Bob have $m$ inputs and $o$ outputs each, such that the value $\mathcal{I}(\{p(a,b\vert x,y)\})=\beta$ self-tests the reference experiment $R=\{\ket{\psi'},\M'_{a\vert x},\N'_{b\vert y}\}$, where $\M'_{a\vert x},\N'_{b\vert y}$ are rank-one projectors and $p(a,b\vert x,y)>0~\forall~a,b$ for each combination $(x,y)$ appearing in $\mathcal{I}$. 
Consider the scenario of Fig. 1b with $m$ inputs and $o^n$ outputs per party. Then, there exists a collection of $n$ non-linear Bell expressions $\mathcal{J}^{i}$ ($i=1,...,n$) for this scenario that self-test the reference experiment $R_n=\{\ket{\psi'}^{\tp n},\otimes_{i=1}^n\M'_{a_i\vert x},\otimes_{i=1}^n\N'_{b_i\vert y}\}$.
\end{theorem}
The nonlinear Bell expressions $\mathcal{J}^{i}$ in Theorem \ref{theorem1} are constructed as follows. Define  $\mathcal{I}^{1}=\mathcal{I}(\{p(a_1,b_1\vert x,y)\})$ and the conditional Bell expressions for the pair $i>1$ as 
\begin{align}\label{Conditional}
    \mathcal{I}_{\mathbf{a}_{i-1}\mathbf{b}_{i-1}}^{i}=\mathcal{I}(\{p(a_{i},b_{i}\vert x,y,\mathbf{a}_{i-1},\mathbf{b}_{i-1}\}),
\end{align}
where $\mathbf{a}_{i-1}=(a_1,a_2,\cdots,a_{i-1})$ and $\mathbf{b}_{i-1}=(b_1,b_2,\cdots,b_{i-1})$.  $\mathcal{I}_{\mathbf{a}_{i-1}\mathbf{b}_{i-1}}^{i}$ gives the value of $\mathcal{I}$ for the pair $i$ conditioned on observing the particular values $\mathbf{a}_{i-1}=(a_1,a_2,\cdots,a_{i-1})$ and $\mathbf{b}_{i-1}=(b_1,b_2,\cdots,b_{i-1})$. Note that in order for these conditional Bell expressions to be well defined we require that $p(\mathbf{a}_{i-1},\mathbf{b}_{i-1}\vert x,y)>0~\forall~\mathbf{a}_{i-1},\mathbf{b}_{i-1}$ for each combination $(x,y)$ appearing in $\mathcal{I}$. This is automatically the case for the reference experiment $R_n$ due to the properties of $R$ given in the statement of the theorem. The Bell expression $\mathcal{J}^{i}$ is defined as 
\begin{align}\label{definitionIJ}
    \mathcal{J}^{i}=\frac{1}{o^{2(i-1)}}\sum_{\mathbf{a}_{i-1}\mathbf{b}_{i-1}}\mathcal{I}_{\mathbf{a}_{i-1}\mathbf{b}_{i-1}}^{i}, 
\end{align}
for $i>1$ and $\mathcal{J}^{1}=\mathcal{I}^{1}$. The observation that $\mathcal{J}^{i}=\beta~\forall i=1,\cdots,n$ self tests $R_n$. The proof is inductive (given a self-test of $k$ copies, the value of $\mathcal{J}^{k+1}$ self-tests an additional copy) and is given in the Appendix \ref{apptheorem}. Here we provide some intuition behind the proof. As shown on the example of the CHSH inequality two ingredients are necessary to make a self-testing statement with the minimal number of inputs. First, all sets of probabilities $p(a_i,b_i|x,y)$ reach the maximum of $\mathcal{I}$, and the pairs of outputs $(a_i,b_i)$ are mutually independent. If $\mathcal{J}^i$ takes value $\beta$ for all $i$, it means that all conditional Bell expressions $\mathcal{I}_{\mathbf{a}_{i-1}\mathbf{b}_{i-1}}^{i}$ are also equal to $\beta$. It can be shown that the maximal values of conditional Bell expressions imply that the outputs corresponding to different $i$-s are mutually independent. In other words, conditioning on $\mathbf{a}_{i-1},\mathbf{b}_{i-1}$ does not affect the values of $a_i$ and $b_i$. Hence, $\beta = \mathcal{I}(p(a_i,b_i|x,y,\mathbf{a}_{i-1},\mathbf{b}_{i-1}) = \mathcal{I}(p(a_i,b_i|x,y)$, for all $i$-s, which means that both conditions for making a self-testing statement are met.

A direct consequence of Theorem \ref{theorem1} is the possibility of self-testing $n$ copies of the two-qubit maximally entangled state $\ket{\phi^+}$ (itself a maximally entangled state of local dimension $2^n$) with only two measurement settings per party via the CHSH Bell inequality. More precisely, we have the following corollary:
\begin{cor}\label{cor:maxent}
Let $\M_{a\vert x},\N_{b\vert y}$ ($x,y \in \{0,1\}, a,b\in\{\pm1\}$) be the local measurements that lead to the maximal violation of the CHSH Bell inequality when applied to $\ket{\phi^+}$. Then the correlations obtained by performing the experiment $R=\{\ket{\phi^+},\M_{a\vert x},\N_{b\vert y}\}$ in the parallel scheme of Fig.\ \ref{ParST}c self-test the reference experiment $\{\ket{\phi^+}^{\tp n},\M_{a\vert x}^{\tp n},\N_{b\vert y}^{\tp n}\}$
\end{cor}
\noindent In particular, this means that an unbounded amount of entanglement can be certified in a device-independent manner with the minimum number of local measurements possible. 

Although Theorem \ref{theorem1}  holds only for the case of perfect statistics, one can investigate the robustness to noise of Corollary \ref{cor:maxent} for the case $n=2$, via the technique propsed in \cite{PhysRevA.91.022115,PhysRevLett.113.040401}. The precise noise model we consider is one in which each copy of the state is subject to the same level of white noise. That is, the observed correlations are generated using the same measurement strategy on the state $\rho_\nu\tp\rho_\nu$ where
\begin{align}
    \rho_{\nu}=\nu\proj{\phi^+}+(1-\nu)\openone/4.
\end{align}
In \cite{PhysRevA.91.022115,PhysRevLett.113.040401} a semi-definite program is given that calculates a value $f$, such that for any state $\rho_{\rA\rB}$ leading to the observed correlations, there exists a local transformation mapping $\rho_{\rA\rB}$ to a state that has fidelity at least $f$ with the reference state (in this case $\ket{\phi^+}\tp\ket{\phi^+})$. In the case of perfect self-testing \eqref{eq:oneselftest}, one has $f=1$, which then decreases as a function of the noise parameter. Fig.\ \ref{fig:plot} shows the values of $f$ obtained as a function of $\nu$ for this noise model.\\

\begin{figure}
\centering
\includegraphics[width=0.9\columnwidth]{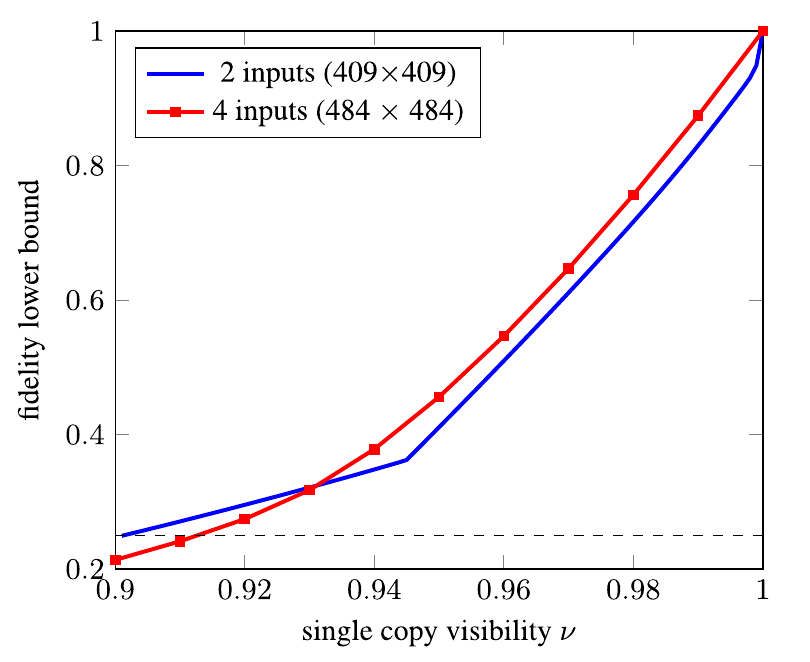}
\caption{\label{fig:plot} Lower bounds to the self-tested fidelity with the state $\ket{\phi^+}\otimes\ket{\phi^+}$ obtained via the method proposed in Refs. \cite{PhysRevA.91.022115,PhysRevLett.113.040401}, using correlations obtained from two copies of the isotropic state $\nu\proj{\phi^+}+(1-\nu)\openone/4$. The blue curve is obtained using two inputs per party, using the parallel scheme of Fig.\ \ref{ParST}c. The red curve is data taken from \cite{Wu2016} using the strategy of Fig.\ \ref{ParST}b with four inputs per party. The legend gives the size of the matrices used in the corresponding SDP optimization. The dashed line corresponds to the trivial fidelity between the target state and the separable state $\ket{00}$ of 0.25. 
}
\end{figure}

\emph{Extensions }-- 
Theorem \ref{theorem1} is defined for bipartite Bell inequalities that feature an equal number of inputs per party. Here we present some more general extensions to this result that either do not require the knowledge of a specific Bell inequality, or allow one to self-test tensor products of different bipartite states. 

\emph{Parallel self-testing protocols from full statistics}---While Theorem \ref{theorem1} refers to self-tests based on the maximal violation of some Bell inequality, it is worth noting that similar claims can be made for self-testing protocols based on the observation of a particular set of correlations. More precisely, we have the following result. 
\begin{theorem}\label{thm:fullstat}
Assume that the correlations $\{p(a,b\vert x,y)\}$ self-test the reference state $\ket{\psi'}$, such that $p(a,b\vert x,y)>0$ for all $a,b,x,y$. Then the correlations $\{p(\mathbf{a},\mathbf{b}\vert x,y)\}$ obtained by performing the reference experiment $n$ times in the parallel scheme of Fig.\ \ref{ParST}c self-test the corresponding tensor product $\ket{\psi'}^{\otimes n}$.
\end{theorem}
In Appendix \ref{appexample} we give an example of this technique that self-tests tensor products of the maximally entangled two-qubit state based on the correlations presented in \cite{1367-2630-18-2-025021}. 

\emph{Combining self-testing protocols}---The previous results show how it is possible to self-test a tensor product $\ket{\psi'}^{\otimes n}$ of the same state $\ket{\psi'}$, starting from a given protocol to self-test the state $\ket{\psi'}$. In fact, the following describes how our construction can also be used to generate self-testing protocols for a tensor product of different states $\bigotimes_{i}\ket{\psi_i'}$.
\begin{theorem}
Consider a set of states $\{\ket{\psi_i'}\}$, each of which is self-tested through the maximal value of a corresponding Bell expression $\mathcal{I}_i$, where each $\mathcal{I}_i$ features the same number of inputs for each party. Then the correlations $\{p(\mathbf{a},\mathbf{b} \vert x,y)\}$ obtained by performing the reference experiments $n$ times in the parallel scheme of Fig.\ \ref{ParST}c self-test the state $\bigotimes_{i}\ket{\psi_i'}$. 
\end{theorem}
The proof of this is almost identical to the proof of Theorem \ref{theorem1}, instead now in \eqref{Conditional} the Bell expressions $\mathcal{I}_i$ are different for each $i$; see Appendix \ref{comb} for more details. One implication of this is the possibility to self-test a tensor product of different partially entangled pairs of qubits with two inputs for each party, which can be achieved through the use of different tilted CHSH inequalities \cite{Bamps}. 

\emph{All self-testing protocols can be paralellized}--- Here we discuss conditions for parallel self-testing without aiming to keep the number of inputs constant. As mentioned in the introduction, a number of schemes to parallel self-test certain classes of states are known (see \cite{Coladangelo2017a,McKague2017a,Coudron2016,2058-9565-3-1-015002,Diagram}), where the number of inputs scales exponentially with the number of copies to be self-tested. Given these examples, one interesting question is thus whether any self-testing protocol for a state $\ket{\psi}$ can be `parallelized' to self-test the tensor product  $\ket{\psi}^{\otimes n}$, at the expense of having exponentially many inputs. We give a positive answer to this question, answering an even a more general question of self-testing a tensor product $\bigotimes_{i}\ket{\psi_i}$ of different states,  given that there are self-tests for the individual states $\ket{\psi_i}$. Note that unlike the previous results, no restrictions are needed on the reference correlations (i.e. some probabilities can be zero). 

\begin{theorem}\label{theorem2}
Consider a set of $n$ bipartite Bell expressions $\{\mathcal{I}_i\}$ characterised by $m_i$ inputs and $o_i$ outputs, respectively, such that the value $\mathcal{I}_i(\{p(a_i,b_i\vert x_i,y_i)\})=\beta_i$ self-tests the reference experiment $R_i = \{\ket{\psi_i'},\M'_{a_i\vert x_i},\N'_{b_i\vert y_i}\}$, where $\M'_{a_i\vert x_i},\N'_{b_i\vert y_i}$ are rank-one projectors for each $i$. Then, the correlations obtained by performing the $R_i$'s in parallel as in Fig.\ 1b self-test the reference experiment $R_n=\{\bigotimes_{i=1}\ket{\psi_i'}, \bigotimes_i\M'_{a_i\vert x},\otimes_{i=1}^n\N'_{b_i\vert y}\}$.
\end{theorem}

The result is proven by constructing a single Bell expression whose maximal value self-tests $R_n$. For each pair $i$ define the Bell value conditioned on particular choice of other inputs $x_j,y_j$, $j\neq i$:
\begin{equation}
    \mathcal{I}^i_{\xX_{(i)}\yY_{(i)}} = \sum_{\aA_{(i)},\bB_{(i)}}\mathcal{I}_i(p(\aA,\bB|\xX,\yY))
\end{equation}
where $\aA_{(i)} = \{a_1,\cdots a_{i-1},a_{i+1},\cdots, a_n\}$ and similarly for  $\bB_{(i)}$, $\xX_{(i)}$ and $\yY_{(i)}$.  The Bell expression $\mathcal{I}_i$ averaged over all other inputs $\xX_{(i)}$ and $\yY_{(i)}$ is 
\begin{equation}
   \mathcal{J}^{i} = \frac{1}{\prod_{j\neq i}m_j^2} \sum_{\mathbf{x}_{(i)}\mathbf{y}_{(i)}} \mathcal{I}^i_{\xX_{(i)}\yY_{(i)}}.
\end{equation}
If $\mathcal{J}^{i}$ achieves its maximal value, it means that for for every  $\xX_i$ and $\yY_i$ $ \mathcal{I}^i_{\xX_i\yY_i}$ also achieves the maximal value, which is enough to prove self-testing of the reference experiment $R_n$. The proof of the theorem is given in Appendix \ref{genparallel}. 

\emph{Application: unbounded randomness expansion.}-- Self-testing is intrinsically related to  device-independent randomness expansion \cite{Coudron:2014:IRE:2591796.2591873,expansion1,expansion2}. This is because, once we certify that the system is in the state $\ket{\psi}$, we can also conclude that any external system is uncorrelated to it. Thus, an external observer can not predict the outcomes of the measurements applied to the system of interest, \ie the outcomes are random. In particular, if the state is maximally entangled of local dimension $d$, and the measurements applied to it are rank-one projective, the amount of random bits obtained per round is $\log_2 (d)$. Notice, however, that in a Bell test some initial amount of randomness must be consumed in the choice of inputs. Thus, the efficiency of the randomness generation protocol is typically given by the the trade-off between the initial randomness consumed and the final randomness obtained. Our Theorem \ref{theorem1} applied to the CHSH inequality and self-testing $n$ copies of a maximally entangled state shows that the best trade-off can be achieved from the noiseless correlations, where only one bit of randomness is used to generate $\log_2(d)$ bits or randomness per round. We thus suspect that our results could improve extraction rates of practical device-independent randomness expansion protocols, although this would require a closer look at issues of noise robustness. 

\emph{Discussion}--In this manuscript we introduced a new procedure useful in parallel self-testing. It allows to self-test highly entangled quantum states with a constant number of measurements. Such schemes are important in protocols for randomness expansion: a small amount of randomness can be expanded to a string of unbounded length. At the heart of our construction lies an interesting insight: independent Bell violations can be used to ensure independence of sources even when the measurements are perfectly correlated. There are several directions for future research on the topic. It would be interesting to explore noise robustness of our scheme. Furthermore, the condition for self-testing can be seen as the maximal violation of a non-linear Bell inequality. One might try to understand if this can be achieved using a single linear Bell inequality.

\emph{Note}-- While working on this project we became aware of the work \cite{warsaw} exploring self-testing of quantum systems of arbitrary local dimension with minimal number of measurements.

\emph{Acknowledgements}--- We thank Jean-Daniel Bancal for sharing numerical data, to Jed Kaniewski and Yeong-Cherng Liang for useful comments and suggestions and to the anonymous referee for many useful comments. We acknowledge funding from SNSF (grant DIAQ), Ramon y Cajal fellowship, the Government of Spain (FIS2020-TRANQI, Retos Quspin and Severo Ochoa CEX2019-000910-S), Fundació Cellex, Fundació Mir-Puig, Generalitat de Catalunya (CERCA, AGAUR SGR 1381),ERC AdG CERQUTE and the AXA Chair in Quantum Information Science.

\onecolumngrid
\bibliographystyle{alphaurl}
\bibliography{sample}

\newcommand{\etalchar}[1]{$^{#1}$}
\begin{thebibliography}{SMSC{\etalchar{+}}15}

\bibitem[ABG{\etalchar{+}}07]{acin2007device}
Antonio Ac{\'\i}n, Nicolas Brunner, Nicolas Gisin, Serge Massar, Stefano
  Pironio, and Valerio Scarani.
\newblock Device-independent security of quantum cryptography against
  collective attacks.
\newblock {\em Physical Review Letters}, 98(23):230501, 2007.
\newblock \href {http://dx.doi.org/10.1103/PhysRevLett.98.230501}
  {\path{doi:10.1103/PhysRevLett.98.230501}}.

\bibitem[BCP{\etalchar{+}}14]{bellreview}
Nicolas Brunner, Daniel Cavalcanti, Stefano Pironio, Valerio Scarani, and
  Stephanie Wehner.
\newblock Bell nonlocality.
\newblock {\em Rev. Mod. Phys.}, 86:419--478, Apr 2014.
\newblock \href {http://dx.doi.org/10.1103/RevModPhys.86.419}
  {\path{doi:10.1103/RevModPhys.86.419}}.

\bibitem[Bel64]{bell}
John~Stewart Bell.
\newblock On the {E}instein-{P}odolsky-{R}osen paradox.
\newblock {\em Physics}, 1:195--200, 1964.
\newblock URL: \url{https://cds.cern.ch/record/111654}.

\bibitem[BKM19]{Diagram}
Spencer Breiner, Amir Kalev, and Carl~A. Miller.
\newblock Parallel self-testing of the {GHZ} state with a proof by diagrams.
\newblock In Peter Selinger and Giulio Chiribella, editors, {\em {\rm
  Proceedings of the 15th International Conference on} Quantum Physics and
  Logic, {\rm Halifax, Canada, 3-7th June 2018}}, volume 287 of {\em Electronic
  Proceedings in Theoretical Computer Science}, pages 43--66. Open Publishing
  Association, 2019.
\newblock \href {http://dx.doi.org/10.4204/EPTCS.287.3}
  {\path{doi:10.4204/EPTCS.287.3}}.

\bibitem[BNS{\etalchar{+}}15]{PhysRevA.91.022115}
Jean-Daniel Bancal, Miguel Navascu\'es, Valerio Scarani, Tam\'as V\'ertesi, and
  Tzyh~Haur Yang.
\newblock Physical characterization of quantum devices from nonlocal
  correlations.
\newblock {\em Phys. Rev. A}, 91:022115, Feb 2015.
\newblock \href {http://dx.doi.org/10.1103/PhysRevA.91.022115}
  {\path{doi:10.1103/PhysRevA.91.022115}}.

\bibitem[BP15]{Bamps}
C\'edric Bamps and Stefano Pironio.
\newblock Sum-of-squares decompositions for a family of
  {C}lauser-{H}orne-{S}himony-{H}olt-like inequalities and their application to
  self-testing.
\newblock {\em Phys. Rev. A}, 91:052111, May 2015.
\newblock \href {http://dx.doi.org/10.1103/PhysRevA.91.052111}
  {\path{doi:10.1103/PhysRevA.91.052111}}.

\bibitem[B{\v{S}}CA18a]{Bowles2018a}
Joseph Bowles, Ivan {\v{S}}upi{\'{c}}, Daniel Cavalcanti, and Antonio
  Ac\'{\i}n.
\newblock Device-independent entanglement certification of all entangled
  states.
\newblock {\em Phys. Rev. Lett.}, 121:180503, Oct 2018.
\newblock \href {http://dx.doi.org/10.1103/PhysRevLett.121.180503}
  {\path{doi:10.1103/PhysRevLett.121.180503}}.

\bibitem[B{\v{S}}CA18b]{BSCA}
Joseph Bowles, Ivan {\v{S}}upi\'{c}, Daniel Cavalcanti, and Antonio Ac\'{\i}n.
\newblock Self-testing of {P}auli observables for device-independent
  entanglement certification.
\newblock {\em Phys. Rev. A}, 98:042336, Oct 2018.
\newblock \href {http://dx.doi.org/10.1103/PhysRevA.98.042336}
  {\path{doi:10.1103/PhysRevA.98.042336}}.

\bibitem[CGJV19]{leash}
Andrea Coladangelo, Alex~B Grilo, Stacey Jeffery, and Thomas Vidick.
\newblock Verifier-on-a-leash: new schemes for verifiable delegated quantum
  computation, with quasilinear resources.
\newblock In {\em Annual International Conference on the Theory and
  Applications of Cryptographic Techniques}, pages 247--277. Springer, 2019.
\newblock \href {http://dx.doi.org/10.1007/978-3-030-17659-4_9}
  {\path{doi:10.1007/978-3-030-17659-4_9}}.

\bibitem[CGS17]{Coladangelo2017b}
Andrea Coladangelo, Koon~Tong Goh, and Valerio Scarani.
\newblock All pure bipartite entangled states can be self-tested.
\newblock {\em Nature Communications}, 8:15485, may 2017.
\newblock \href {http://dx.doi.org/10.1038/ncomms15485}
  {\path{doi:10.1038/ncomms15485}}.

\bibitem[CHSH69]{chsh}
John~F. Clauser, Michael~A. Horne, Abner Shimony, and Richard~A. Holt.
\newblock Proposed experiment to test local hidden-variable theories.
\newblock {\em Phys. Rev. Lett.}, 23:880--884, Oct 1969.
\newblock \href {http://dx.doi.org/10.1103/PhysRevLett.23.880}
  {\path{doi:10.1103/PhysRevLett.23.880}}.

\bibitem[CK11]{expansion1}
Roger Colbeck and Adrian Kent.
\newblock Private randomness expansion with untrusted devices.
\newblock {\em Journal of Physics A: Mathematical and Theoretical},
  44(9):095305, 2011.
\newblock \href {http://dx.doi.org/10.1088/1751-8113/44/9/095305}
  {\path{doi:10.1088/1751-8113/44/9/095305}}.

\bibitem[CN16]{Coudron2016}
Matthew Coudron and Anand Natarajan.
\newblock The parallel-repeated magic square game is rigid, 2016.
\newblock arXiv:1609.06306.
\newblock URL: \url{https://arxiv.org/abs/1609.06306}.

\bibitem[Col17]{Coladangelo2017a}
Andrea Coladangelo.
\newblock Parallel self-testing of (tilted) {EPR} pairs via copies of (tilted)
  {CHSH} and the magic square game.
\newblock {\em Quantum Information and Computation}, 17(9-10):831--865, 2017.
\newblock \href {http://dx.doi.org/10.26421/QIC17.9-10-6}
  {\path{doi:10.26421/QIC17.9-10-6}}.

\bibitem[CRSV18]{Chao2016}
Rui Chao, Ben~W. Reichardt, Chris Sutherland, and Thomas Vidick.
\newblock Test for a large amount of entanglement, using few measurements.
\newblock {\em {Quantum}}, 2:92, September 2018.
\newblock \href {http://dx.doi.org/10.22331/q-2018-09-03-92}
  {\path{doi:10.22331/q-2018-09-03-92}}.

\bibitem[CY14]{Coudron:2014:IRE:2591796.2591873}
Matthew Coudron and Henry Yuen.
\newblock Infinite randomness expansion with a constant number of devices.
\newblock In {\em Proceedings of the Forty-sixth Annual ACM Symposium on Theory
  of Computing}, STOC '14, pages 427--436, New York, NY, USA, 2014. ACM.
\newblock \href {http://dx.doi.org/10.1145/2591796.2591873}
  {\path{doi:10.1145/2591796.2591873}}.

\bibitem[GVW{\etalchar{+}}15]{loopholefree3}
Marissa Giustina, Marijn A.~M. Versteegh, S\"oren Wengerowsky, Johannes
  Handsteiner, Armin Hochrainer, Kevin Phelan, Fabian Steinlechner, Johannes
  Kofler, Jan-\AA{}ke Larsson, Carlos Abell\'an, Waldimar Amaya, Valerio
  Pruneri, Morgan~W. Mitchell, J\"orn Beyer, Thomas Gerrits, Adriana~E. Lita,
  Lynden~K. Shalm, Sae~Woo Nam, Thomas Scheidl, Rupert Ursin, Bernhard
  Wittmann, and Anton Zeilinger.
\newblock Significant-loophole-free test of bell's theorem with entangled
  photons.
\newblock {\em Phys. Rev. Lett.}, 115:250401, Dec 2015.
\newblock \href {http://dx.doi.org/10.1103/PhysRevLett.115.250401}
  {\path{doi:10.1103/PhysRevLett.115.250401}}.

\bibitem[HBD{\etalchar{+}}15]{loopholefree1}
B.~Hensen, H.~Bernien, A.~E. Dréau, A.~Reiserer, N.~Kalb, M.~S. Blok,
  J.~Ruitenberg, R.~F.~L. Vermeulen, R.~N. Schouten, C.~Abellán, and et~al.
\newblock Loophole-free bell inequality violation using electron spins
  separated by 1.3 kilometres.
\newblock {\em Nature}, 526(7575):682–686, 2015.
\newblock \href {http://dx.doi.org/10.1038/nature15759}
  {\path{doi:10.1038/nature15759}}.

\bibitem[JMS20]{parallelQKD}
Rahul Jain, Carl~A. Miller, and Yaoyun Shi.
\newblock Parallel device-independent quantum key distribution.
\newblock {\em IEEE Transactions on Information Theory}, 66(9):5567–5584, Sep
  2020.
\newblock \href {http://dx.doi.org/10.1109/tit.2020.2986740}
  {\path{doi:10.1109/tit.2020.2986740}}.

\bibitem[KM18]{2058-9565-3-1-015002}
Amir Kalev and Carl~A Miller.
\newblock Rigidity of the magic pentagram game.
\newblock {\em Quantum Science and Technology}, 3(1):015002, 2018.
\newblock \href {http://dx.doi.org/10.1088/1367-2630/18/2/025021}
  {\path{doi:10.1088/1367-2630/18/2/025021}}.

\bibitem[LLR{\etalchar{+}}21]{expansion2}
Wen-Zhao Liu, Ming-Han Li, Sammy Ragy, Si-Ran Zhao, Bing Bai, Yang Liu, Peter~J
  Brown, Jun Zhang, Roger Colbeck, Jingyun Fan, et~al.
\newblock Device-independent randomness expansion against quantum side
  information.
\newblock {\em Nature Physics}, pages 1--4, 2021.
\newblock \href {http://dx.doi.org/10.1038/s41567-020-01147-2}
  {\path{doi:10.1038/s41567-020-01147-2}}.

\bibitem[McK17]{McKague2017a}
Matthew McKague.
\newblock Self-testing in parallel with {CHSH}.
\newblock {\em {Quantum}}, 1:1, April 2017.
\newblock \href {http://dx.doi.org/10.22331/q-2017-04-25-1}
  {\path{doi:10.22331/q-2017-04-25-1}}.

\bibitem[MY04]{Mayers2004}
Dominic Mayers and Andrew Yao.
\newblock Self testing quantum apparatus.
\newblock {\em Quantum Info. Comput.}, 4:273, 2004.
\newblock \href {http://dx.doi.org/10.26421/QIC4.4-3}
  {\path{doi:10.26421/QIC4.4-3}}.

\bibitem[NPA07]{npa1}
Miguel Navascu\'es, Stefano Pironio, and Antonio Ac\'{\i}n.
\newblock Bounding the set of quantum correlations.
\newblock {\em Phys. Rev. Lett.}, 98:010401, Jan 2007.
\newblock \href {http://dx.doi.org/10.1103/PhysRevLett.98.010401}
  {\path{doi:10.1103/PhysRevLett.98.010401}}.

\bibitem[NPA08]{npa2}
Miguel Navascu{\'{e}}s, Stefano Pironio, and Antonio Ac{\'{i}}n.
\newblock A convergent hierarchy of semidefinite programs characterizing the
  set of quantum correlations.
\newblock {\em New Journal of Physics}, 10(7):073013, 2008.
\newblock \href {http://dx.doi.org/10.1088/1367-2630/10/7/073013}
  {\path{doi:10.1088/1367-2630/10/7/073013}}.

\bibitem[NV17]{Natarajan}
Anand Natarajan and Thomas Vidick.
\newblock A quantum linearity test for robustly verifying entanglement.
\newblock In {\em Proceedings of the 49th Annual ACM SIGACT Symposium on Theory
  of Computing}, STOC 2017, pages 1003--1015, New York, NY, USA, 2017. ACM.
\newblock \href {http://dx.doi.org/10.1145/3055399.3055468}
  {\path{doi:10.1145/3055399.3055468}}.

\bibitem[NV18]{Natarajan2018}
Anand {Natarajan} and Thomas {Vidick}.
\newblock Low-degree testing for quantum states, and a quantum entangled games
  pcp for qma.
\newblock In {\em 2018 IEEE 59th Annual Symposium on Foundations of Computer
  Science (FOCS)}, pages 731--742, Oct 2018.
\newblock \href {http://dx.doi.org/10.1109/FOCS.2018.00075}
  {\path{doi:10.1109/FOCS.2018.00075}}.

\bibitem[OV16]{Ostrev2016}
Dimiter Ostrev and Thomas Vidick.
\newblock The structure of nearly-optimal quantum strategies for the {CHSH} (n)
  {XOR} games.
\newblock {\em Quantum Information \& Computation}, 16(13-14),
  pp.(13-14):1191--1211, 2016.
\newblock \href {http://dx.doi.org/10.26421/QIC16.13-14-6}
  {\path{doi:10.26421/QIC16.13-14-6}}.

\bibitem[PAM{\etalchar{+}}10]{pironio2010random}
Stefano Pironio, Antonio Ac{\'\i}n, Serge Massar, A~Boyer de~La~Giroday,
  Dzimitry~N Matsukevich, Peter Maunz, Steven Olmschenk, David Hayes, Le~Luo,
  T~Andrew Manning, et~al.
\newblock Random numbers certified by {B}ell’s theorem.
\newblock {\em Nature}, 464(7291):1021--1024, 2010.
\newblock \href {http://dx.doi.org/10.1038/nature09008}
  {\path{doi:10.1038/nature09008}}.

\bibitem[PR92]{Popescu1992}
Sandu Popescu and Daniel Rohrlich.
\newblock Which states violate {B}ell's inequality maximally?
\newblock {\em Physics Letters A}, 169(6):411--414, 1992.
\newblock \href {http://dx.doi.org/10.1016/0375-9601(92)90819-8}
  {\path{doi:10.1016/0375-9601(92)90819-8}}.

\bibitem[RUV13]{ruv}
Ben~W. Reichardt, Falk Unger, and Umesh Vazirani.
\newblock Classical command of quantum systems.
\newblock {\em Nature}, 496:456, 2013.
\newblock \href {http://dx.doi.org/10.1038/nature12035}
  {\path{doi:10.1038/nature12035}}.

\bibitem[SMSC{\etalchar{+}}15]{loopholefree2}
Lynden~K. Shalm, Evan Meyer-Scott, Bradley~G. Christensen, Peter Bierhorst,
  Michael~A. Wayne, Martin~J. Stevens, Thomas Gerrits, Scott Glancy, Deny~R.
  Hamel, Michael~S. Allman, Kevin~J. Coakley, Shellee~D. Dyer, Carson Hodge,
  Adriana~E. Lita, Varun~B. Verma, Camilla Lambrocco, Edward Tortorici, Alan~L.
  Migdall, Yanbao Zhang, Daniel~R. Kumor, William~H. Farr, Francesco Marsili,
  Matthew~D. Shaw, Jeffrey~A. Stern, Carlos Abell\'an, Waldimar Amaya, Valerio
  Pruneri, Thomas Jennewein, Morgan~W. Mitchell, Paul~G. Kwiat, Joshua~C.
  Bienfang, Richard~P. Mirin, Emanuel Knill, and Sae~Woo Nam.
\newblock Strong loophole-free test of local realism.
\newblock {\em Phys. Rev. Lett.}, 115:250402, Dec 2015.
\newblock \href {http://dx.doi.org/10.1103/PhysRevLett.115.250402}
  {\path{doi:10.1103/PhysRevLett.115.250402}}.

\bibitem[SSKA19]{warsaw}
Shubhayan Sarkar, Debashis Saha, Jedrzej Kaniewski, and Remigiusz Augusiak.
\newblock Self-testing quantum systems of arbitrary local dimension with the
  minimal number of measurements, 2019.
\newblock \href {http://arxiv.org/abs/1909.12722} {\path{arXiv:1909.12722}}.

\bibitem[SW87]{Summers1987}
Stephen~J. Summers and Reinhard~F. Werner.
\newblock Maximal violation of {B}ell's inequalities is generic in quantum
  field theory.
\newblock {\em Communications in Mathematical Physics}, 110(2):247--259, 1987.
\newblock \href {http://dx.doi.org/10.1007/BF01207366}
  {\path{doi:10.1007/BF01207366}}.

\bibitem[Tsi93]{Tsirelson1993}
Boris Tsirelson.
\newblock Some results and problems on quantum {B}ell-type inequalities.
\newblock {\em Hadronis Journal Supplement}, 8:329--45, 1993.
\newblock URL: \url{https://ci.nii.ac.jp/naid/10026857475/en/}.

\bibitem[vB20]{STreview}
Ivan \v{S}upi\'{c} and Joseph Bowles.
\newblock Self-testing of quantum systems: a review.
\newblock {\em Quantum}, 4:337, Sep 2020.
\newblock \href {http://dx.doi.org/10.22331/q-2020-09-30-337}
  {\path{doi:10.22331/q-2020-09-30-337}}.

\bibitem[Vid17]{parallelQKD2}
Thomas Vidick.
\newblock {P}arallel {DIQKD} from parallel repetition, 2017.
\newblock \href {http://arxiv.org/abs/1703.08508} {\path{arXiv:1703.08508}}.

\bibitem[WBMS16]{Wu2016}
Xingyao Wu, Jean-Daniel Bancal, Matthew McKague, and Valerio Scarani.
\newblock Device-independent parallel self-testing of two singlets.
\newblock {\em Phys. Rev. A}, 93:062121, 2016.
\newblock \href {http://dx.doi.org/10.1103/PhysRevA.93.062121}
  {\path{doi:10.1103/PhysRevA.93.062121}}.

\bibitem[WWS16]{1367-2630-18-2-025021}
Yukun Wang, Xingyao Wu, and Valerio Scarani.
\newblock All the self-testings of the singlet for two binary measurements.
\newblock {\em New Journal of Physics}, 18(2):025021, 2016.
\newblock \href {http://dx.doi.org/10.1088/1367-2630/18/2/025021}
  {\path{doi:10.1088/1367-2630/18/2/025021}}.

\bibitem[YN13]{Yang}
Tzyh~Haur Yang and Miguel Navascu\'es.
\newblock Robust self-testing of unknown quantum systems into any entangled
  two-qubit states.
\newblock {\em Phys. Rev. A}, 87:050102, May 2013.
\newblock \href {http://dx.doi.org/10.1103/PhysRevA.87.050102}
  {\path{doi:10.1103/PhysRevA.87.050102}}.

\bibitem[YVB{\etalchar{+}}14]{PhysRevLett.113.040401}
Tzyh~Haur Yang, Tam\'as V\'ertesi, Jean-Daniel Bancal, Valerio Scarani, and
  Miguel Navascu\'es.
\newblock Robust and versatile black-box certification of quantum devices.
\newblock {\em Phys. Rev. Lett.}, 113:040401, Jul 2014.
\newblock \href {http://dx.doi.org/10.1103/PhysRevLett.113.040401}
  {\path{doi:10.1103/PhysRevLett.113.040401}}.

\end{thebibliography}

\onecolumngrid
\setcounter{equation}{0}
\setcounter{figure}{0}
\setcounter{table}{0}
\setcounter{section}{0}
\makeatletter
\renewcommand{\theequation}{A\arabic{equation}}
\renewcommand{\thefigure}{A\arabic{figure}}
v

\begin{appendix}

\section{Proof of Theorem \ref{theorem1}}\label{apptheorem}

Before starting the proofs let us introduce some useful notation. The probabilities to observe outcomes $\aA = (a_1,\cdots, a_n)$ and $\bB = (b_1,\cdots, b_n)$ when the inputs are $x$ and $y$ are given by the Born rule:
\begin{equation}
    p(\aA,\bB|x,y) = \Tr[\M_{\aA|x}\tp\N_{\bB|y}\rho^{\mathrm{A}\mathrm{B}}].
\end{equation}
Let us introduce further auxilliary measurement operators
\begin{equation}
    \M_{a|x}^{(i)} = \sum_{\aA,a_i = a}\M_{\aA|x} \qquad \N_{b|y}^{(i)} = \sum_{\bB,b_i = b}\N_{\bB|y}.
\end{equation}
With $\M'_{a|x}$ we simply denote the reference single qudit measurements. Since we extract the tensor product of the reference state into the ancillary Hilbert space, to ease keeping track of the number of the extracted reference states we denote the Hilbert spaces of ancillary systems with $\rA_j$ and $\rB_j$ (instead of $\rA'$ and $\rB'$ in the main text). In order to relax the notation when writing measurement operators we omit the Hilbert space notation. Thus, we employ the following notation: $\M_{a|x}^{(i)} \equiv {\M_{a|x}^{(i)}}^{\rA}$, $\M_{\aA|x} \equiv \M_{\aA|x}^{\rA}$, $\N_{b|y}^{(i)} \equiv {\N_{b|y}^{(i)}}^{\rB}$, $\N_{\bB|y} \equiv \N_{\bB|y}^{\rB}$ for physical measurements and $\M'_{a|x} \equiv {\M'}_{a|x}^{\rA_j}$,  $\N'_{b|y} \equiv {\N'}_{b|y}^{\rB_j}$ for reference measurements.   

We prove the theorem by using mathematical induction. In the first step we prove the base case, \emph{i.e.} that the theorem holds when $n=1$. In the second step we prove the so-called inductive step, saying that if the theorem holds for some natural number $i$ it also holds for $i+1$. The whole theorem is proven by demonstrating the correctness of the base and the inductive step. The validity of the base step holds trivially since the theorem assumes that the Bell inequality under consideration is self-testing. The condition $\mathcal{I}_1 = \beta$  implies the existence of the local unitary $U_1 = U_{\rA\rA_1}\tp U_{\rB\rB_1}\tp\idd_\rP$ such that
\begin{align}\label{eq:state1}
    U_1\ket{\psi}^{\ab}\tp\ket{00}^{\rA_1\rB_1} &= \ket{\xi_1}^{\ab}\tp\ket{\psi'}^{\rA_1\rB_1}\\ \label{eq:mnstate1}
    U_1{\M}_{a|x}^{(1)}\tp{\N}_{b|y}^{(1)}\ket{\psi}^{\ab}\tp\ket{00}^{\rA_1\rB_1} &= \ket{\xi_1}^{\ab}\tp\left({{\M}'}_{a|x}^{\rA_1}\tp{{\N}'}_{b|y}^{\rB_1}\right)\ket{\psi'}^{\rA_1\rB_1}
\end{align}
To start the inductive step assume that the theorem holds for $i$, \emph{i.e.} that $\sum_{\aA_{i-1}\bB_{i-1}}\mathcal{I}_{\aA_{i-1}\bB_{i-1}}^{i} = o^{2(i-1)}\beta$ implies there exist the local unitary $U_i = U_{\rA\rA_1\cdots\rA_i}\tp U_{\rB\rB_1\cdots\rB_i}\tp\idd_\rP$ such that
\begin{align}\label{i-1state}
      U_i\ket{\psi}^{\ab}\tp\ket{00}^{\rA_1\rB_1}\cdots \tp\ket{00}^{\rA_i\rB_i} &=\ket{\xi_{i}}^{\ab} \tp \ket{\psi'}^{\rA_1\rB_1}\tp \cdots\tp \ket{\psi'}^{\rA_i\rB_i},\\ \label{i-1stateandmeas}
    U_i\M_{\aA_i|x}\tp\N_{\bB_i|y}\ket{\psi}^{\ab}\tp\ket{00}^{\rA_1\rB_1}\cdots \tp\ket{00}^{\rA_i\rB_i} &= \ket{\xi_{i}}^{\ab} \bigotimes_{j=1}^i \M'_{a_j|x}\tp\N'_{b_j|y}\ket{\psi'}^{\rA_j\rB_j}.
\end{align}
By summing \eqref{i-1stateandmeas} over $\bB_{i}$ and using the completeness relation $\sum_{\bB_i}{\N}_{{\bB_i}|y} = \idd$ we obtain
\begin{equation*}
    U_i\M_{\aA_i|x}\ket{\psi}^{\ab}\tp\ket{00}^{\rA_1\rB_1}\cdots \tp\ket{00}^{\rA_i\rB_i} = \ket{\xi_{i}}^{\ab} \bigotimes_{j=1}^i \M'_{a_j|x}\tp\idd^{\rB_j}\ket{\psi'}^{\rA_j\rB_j},
\end{equation*}
which can be rewritten as 
\begin{multline}\label{shansha}
    \left(U_{\rA\rA_1\cdots\rA_i}\M_{\aA_i|x}\tp\idd^{\rA_1\cdots\rA_{i}}U_{\rA\rA_1\cdots\rA_i}^\dagger\right) U_i\ket{\psi}^{\ab}\tp\ket{00}^{\rA_1\rB_1}\cdots \tp\ket{00}^{\rA_i\rB_i} = \\ =  \ket{\xi_{i}}^{\ab} \bigotimes_{j=1}^i \M'_{a_j|x}\tp\idd^{\rB_j}\ket{\psi'}^{\rA_j\rB_j}.
\end{multline}
By comparing \eqref{shansha} and \eqref{i-1state} we obtain
\begin{equation}\label{yara}
    U_{\rA\rA_1\cdots\rA_i}\M_{\aA_i|x}\tp\idd^{\rA_1\cdots\rA_{i}}U_{\rA\rA_1\cdots\rA_i}^\dagger = S_{\aA_i|x}^{\rA}\tp {\M'}_{\aA_i|x}^{\rA_1\cdots\rA_i}
\end{equation}
where $S_{\aA_i|x}\ket{\xi_i}=\ket{\xi_i}$ for all $\aA_i$ and $x$. Note that \ref{yara} is correct in case $\tr_{B_1\cdots B_i}\ketbra{\psi'}{\psi'}$ is full rank. We make assumption that this is indeed the case, which is a standard assumption in the self-testing scenarios.  Since $U_{\rA\rA_1\cdots\rA_i}$ preserves the identity the condition
\begin{equation}
\label{brienne}
\sum_{\aA_i}S_{\aA_i|x}^{\rA}\tp {\M'}_{\aA_i|x}^{\rA_1\cdots\rA_i} = \idd^{\rA}\tp\idd^{\rA_1\cdots\rA_i}
\end{equation}
must be satisfied. Since $\sum_{\aA_i}\M'_{\aA_i|x}= \idd$, the condition is satisfied if and only if $S_{\aA_i|x} = \idd$ for all $\aA_i$ and $x$. This can be seen through the simple reasoning. $U_{i,\rA}\equiv U_{\rA\rA_1\cdots\rA_i}$ is unitary and thus $\idd \geq S_{\aA_i|x}^{\rA}\tp {\M'}_{\aA_i|x}^{\rA_1\cdots\rA_i} \geq 0$. Hence, for arbitrary quantum states ${\rho}^{\rA}$ and ${\tau}^{\rA_1\cdots\rA_i}$  it holds $\tr(S_{\aA_i|x}^{\rA}\rho) = \alpha_i$ and $\tr(\M'_{\aA_i|x}\tau) = \beta_i$ where $0 \leq \alpha_i,\beta_i \leq 1$. Eq. \eqref{brienne} implies
$\sum_{i}\alpha_i\beta_i = 1$ and the completeness of $\M'_{\aA_i|x}$ implies $\sum_{i}\beta_i = 1$. Thus, it also holds $\sum_{i}(1-\alpha_i)\beta_i = 0$. The sum of nonnegative numbers is equal to zero if and only if each of them is equal to zero. Since the argumentation must hold for all states it implies $S_{\aA_i|x} = \idd$ for all $\aA_i$ and $x$. Eq. \eqref{yara} reduces to
\begin{equation}\label{euron}
    U_{\rA\rA_1\cdots\rA_i}\M_{\aA_i|x}\tp\idd^{\rA_1\cdots\rA_{i}}U_{\rA\rA_1\cdots\rA_i}^\dagger = \idd^{\rA}\tp {\M'}_{\aA_i|x}^{\rA_1\cdots\rA_i}
\end{equation}
for all $\aA_i$ and $x$.  Analogous conclusion can be obtained for Bob's operators
\begin{equation}\label{belon}
    U_{\rB\rB_1\cdots\rB_i}\N_{\bB_i|x}\tp\idd^{\rB_1\cdots\rB_{i}}U_{\rB\rB_1\cdots\rB_i}^\dagger = \idd^{\rB}\tp {\N'}_{\bB_i|x}^{\rB_1\cdots\rB_i}
\end{equation}
 for all $\bB_i$ and $y$. Note now that ${\M}_{\aA_{i}|x} = \sum_{a_{i+1}} {\M}_{\aA_i,a_{i+1}|x}$, where $\aA_i,a_{i+1}$ is just a different way to write $\aA_{i+1}$, where $(i+1)$-th output is separately written, as to be able to keep track of it. Similar eq. holds for Bob and thus eqs. \eqref{euron} and \eqref{belon} imply:
 \begin{align}\label{kl}
     U_{\rA\rA_1\cdots\rA_i}{\M}_{\aA_i,a_{i+1}|x}^{\rA}\tp\idd^{\rA_1\cdots\rA_i}U_{\rA\rA_1\cdots\rA_i}^\dagger &= \K_{\aA_i,a_{i+1}|x}^{\rA}\tp {\M'}_{\aA_i|x}^{\rA_1\cdots\rA_i}, \\ \label{kl2} U_{\rB\rB_1\cdots\rB_i}{\N}_{\bB_i,b_{i+1}|y}^{\rB}\tp\idd^{\rB_1\cdots\rB_i}U_{\rB\rB_1\cdots\rB_i}^\dagger &= \Ll_{\bB_i,b_{i+1}|y}^{\rB}\tp {\N'}_{\bB_i|y}^{\rB_1\cdots\rB_i},
 \end{align}
where operators $\K_{\aA_i,a_{i+1}|x}$ and $\Ll_{\bB_i,b_{i+1}}$ satisfy
\begin{equation}\label{completeness}
    \sum_{a_{i+1}}\K_{\aA_i,a_{i+1}|x} = \idd, \qquad \sum_{b_{i+1}}\Ll_{\bB_i,b_{i+1}|y} = \idd.
\end{equation}
We prove that the operators on the r.h.s. of equalities \eqref{kl} and \eqref{kl2} have the tensor product structure by reductio ab absurdum. Assume $U_{\rA\rA_1\cdots\rA_i}{\M}_{\aA_i,a_{i+1}|x}^{\rA}\tp\idd^{\rA_1\cdots\rA_i}U_{\rA\rA_1\cdots\rA_i}^\dagger \equiv \tilde{K}_{\mathbf{a}_{i+1}}^{\rA\rA_1\cdots\rA_i}$ are entangled across the bipartition $\rA|\rA_1\cdots\rA_i$. Then the following equation holds:
\begin{equation}\label{covid}
    \sum_{a_{i+1}}\tilde{K}_{\mathbf{a}_{i+1}}^{\rA\rA_1\cdots\rA_i} = \idd^{\rA}\tp {\M'}_{\aA_i|x}^{\rA_1\cdots\rA_i}
\end{equation}
Let us denote $\tr_{\rA}\left[\tilde{K}_{\mathbf{a}_{i+1}}^{\rA\rA_1\cdots\rA_i}\right] = k_{a_{i+1}}\bar{K}_{\mathbf{a}_{i+1}}^{\rA_1\cdots\rA_i}$, where $\bar{K}_{\mathbf{a}_{i+1}}^{\rA_1\cdots\rA_i} \geq 0$. Given that $\tilde{K}_{\mathbf{a}_{i+1}}^{\rA\rA_1\cdots\rA_i}$ is positive for all $\mathbf{a}_{i+1}$, we have that $k_{a_{i+1}} \geq 0$ for all $a_{i+1}$. By tracing out the Hilbert space $\rA$ in eq. \eqref{covid} we obtain
\begin{equation}\label{virus}
    \sum_{a_{i+1}}\frac{k_{a_{i+1}}}{D}\bar{K}_{\mathbf{a}_{i+1}}^{\rA_1\cdots\rA_i} = {\M'}_{\aA_i|x}^{\rA_1\cdots\rA_i},
\end{equation}
where $D$ is the dimension of Hilbert space $\mathcal{H}_\rA$. Since ${\M'}_{\aA_i|x}^{\rA_1\cdots\rA_i}$ is, by assumption, a tensor product of rank-one projectors, eq. \eqref{virus} implies $\bar{K}_{\mathbf{a}_{i+1}}^{\rA_1\cdots\rA_i} = {\M'}_{\aA_i|x}^{\rA_1\cdots\rA_i}$ for all $\aA_i$.  Henceforth, we have $\tr_{\rA}\left[\tilde{K}_{\mathbf{a}_{i+1}}^{\rA\rA_1\cdots\rA_i}\right] = k_{a_{i+1}}{\M'}_{\mathbf{a}_{i+1}}^{\rA_1\cdots\rA_i}$, and as ${\M'}_{\mathbf{a}_{i+1}}^{\rA_1\cdots\rA_i}$ is a rank-one projector, there cannot be entanglement in $\tilde{K}_{\mathbf{a}_{i+1}}^{\rA\rA_1\cdots\rA_i}$ across the bipartition $\rA|\rA_1\cdots\rA_i$, which justifies the form of eq. \eqref{kl}. Similar argument holds for \eqref{kl2}.

Let us now write down the expression for the conditional Bell value 
\begin{equation}\label{jedanCov}
 \mathcal{I}_{\aA_i,\bB_i}^{i+1} = \sum_{a_{i+1}b_{i+1}xy}\frac{b_{a_{i+1}b_{i+1}}^{xy}}{p(\aA_i,\bB_i|xy)}\tr\left[\left({\M}_{\aA_i,a_{i+1}|x}\otimes{\N}_{\bB_i,b_{i+1}|y}\right)\rho_{\rA\rB}\tp\ketbra{00}{00}^{\rA_1\rB_1\cdots\rA_i\rB_i}\right],
\end{equation}
where $b_{a_{i+1}b_{i+1}}^{xy}$ are the coefficients from the original Bell inequality corresponding to the probability $p(a_{i+1}b_{i+1}|xy)$.
Given the self-testing statements \eqref{i-1state},\eqref{kl},\eqref{kl2} it can be rewritten in the following way
\begin{align}\label{dvaCov}
 \mathcal{I}_{\aA_i,\bB_i}^{i+1} = \sum_{a_{i+1}b_{i+1}xy}\frac{b_{a_{i+1}b_{i+1}}^{xy}}{p(\aA_i,\bB_i|xy)}\tr\left[\left({\K}_{\aA_i,a_{i+1}|x}\otimes{\Ll}_{\bB_ib_{i+1}|y}\right)\proj{\xi_i}^{\ab}\bigotimes_{j=1}^i\left(\left({\M'}_{\aA_i|x}\otimes{\N'}_{\bB_i|y}\right)\ketbra{\psi'}{\psi'}^{\rA_j\rB_j}\right)\right].
\end{align}
%
The expression can be further simplified
\begin{align}\label{sedam}
\mathcal{I}_{\aA_i,\bB_i}^{i+1}  
 &= \sum_{a_{i+1}b_{i+1}xy}\frac{b_{a_{i+1}b_{i+1}}^{xy}}{p(\aA_i,\bB_i|xy)}\tr\left[\left({\K}_{\aA_i,a_{i+1}|x}\otimes{\Ll}_{\bB_i,b_{i+1}|y}\right)\proj{\xi_i}^{\ab}\right]\prod_{j=1}^i\tr\left[\left({\M'}_{a_j|x}\otimes{\N'}_{b_j|y}\right)\ketbra{\psi'}{\psi'}^{\rA_j\rB_j}\right] \\ \label{osam}
&= \sum_{a_{i+1}b_{i+1}xy}b_{a_{i+1}b_{i+1}}^{xy}\tr\left[\left({\K}_{\aA_i,a_{i+1}|x}\otimes{\Ll}_{\bB_i,b_{i+1}|y}\right)\proj{\xi_i}^{\ab}\right].
\end{align}
The equalities come from the property of trace $\tr(A\tp B) =\tr(A)\tr(B)$ and observing that $p(\aA_i,\bB_i|xy) = \prod_{j=1}^i\tr\left[\left({\M'}_{a_j|x}\otimes{\N'}_{b_j|y}\right)\ketbra{\psi'}{\psi'}^{\rA_j\rB_j}\right] $. Let us now sum different conditional Bell values
\begin{align}\nonumber
\sum_{\aA_i\bB_i}\mathcal{I}_{\aA_i\bB_i}^{i+1} &=  \sum_{\aA_i\bB_ia_{i+1}b_{i+1}xy}b_{a_{i+1}b_{i+1}}^{xy}\tr\left[\left({\K}_{\aA_i,a_{i+1}|x}\otimes{\Ll}_{\bB_i,b_{i+1}|y}\right)\proj{\xi_i}^{\ab}\right] \\ \label{devet} &= o^{2i}\sum_{a_{i+1}b_{i+1}xy}b_{a_{i+1}b_{i+1}}^{xy}\tr\left[\left({\K}^{(i+1)}_{a_{i+1}|x}\otimes{\Ll}^{(i+1)}_{b_{i+1}|y}\right)\proj{\xi_i}^{\ab}\right],
\end{align}
where we introduced new operators
\begin{equation}\label{doublecompleteness}
    o^i{\K}^{(i+1)}_{a_{i+1}|x} = \sum_{\aA_i}{\K}_{\aA_i,a_{i+1}|x}, \qquad 
    o^i{\Ll}^{(i+1)}_{b_{i+1}|y} = \sum_{\bB_i}{\Ll}_{\bB_i,b_{i+1}|y}.
\end{equation}
Note that ${\K}^{(i+1)}_{a_{i+1}|x}$ and ${\Ll}^{(i+1)}_{b_{i+1}|y}$ are positive and satisfy the completeness relations $\sum_{a_{i+1}}{\K}^{(i+1)}_{a_{i+1}|x} = \idd$ and $\sum_{b_{i+1}}{\Ll}^{(i+1)}_{b_{i+1}|y} = \idd$ (see eq. \eqref{completeness} ). Hence they represent valid quantum measurements. The condition for self-testing is $\mathcal{J}^i = \beta$ for all $i$, which through the definition of $\mathcal{J}^i$ given in eq. \eqref{definitionIJ} implies  $\sum_{\aA_i\bB_i}\mathcal{I}_{\aA_i\bB_i}^{i+1} = o^{2i}\beta$, or equivalently
\begin{equation}\label{bell2}
    \sum_{a_{i+1}b_{i+1}xy}b_{a_{i+1}b_{i+1}}^{xy}\tr\left[\left({\K}^{(i+1)}_{a_{i+1}|x}\otimes{\Ll}^{(i+1)}_{b_{i+1}|y}\right)\proj{\xi_i}^{\ab}\right] = \beta.
\end{equation}
Since the Bell inequality is self-testing the reference experiment the eq. \eqref{bell2} implies the existence of the local unitary transformation $U'_{i+1} = U'_{\rA\rA_{i+1}}\tp U'_{\rB\rB_{i+1}}\tp\idd_\rP$ such that
\begin{align}\label{eq:state2aux}
    U'_{i+1}\ket{\xi_i}^{\ab}\tp\ket{00}^{\rA_{i+1}\rB_{i+1}} &= \ket{\xi_{i+1}}^{\ab}\tp\ket{\psi'}^{\rA_{i+1}\rB_{i+1}}\\ \label{eq:mnstate2aux}
    U'_{i+1}\left[{\K}_{a_{i+1}|x}^{(i+1)}\tp{\Ll}_{b_{i+1}|y}^{(i+1)}\ket{\xi_i}^{\ab}\tp\ket{00}^{\rA_{i+1}\rB_{i+1}}\right] &= \ket{\xi_{i+1}}^{\ab}\tp\left({\M'}_{a_{i+1}|x}^{\rA_{i+1}}\tp{\N'}_{b_{i+1}|y}^{\rB_{i+1}}\right)\ket{\psi'}^{\rA_{i+1}\rB_{i+1}}
\end{align}
Combining eqs \eqref{eq:state2aux}-\eqref{eq:mnstate2aux} with \eqref{i-1state}-\eqref{i-1stateandmeas} leads to the parallel self-testing of $\bigotimes_{j=1}^{i+1}\ket{\psi'}$:
\begin{align}\label{eq:state2}
    U'_{i+1}\left[U_i\left[\ket{\psi}^{\ab}\tp\ket{00}^{\rA_1\rB_1\cdots\rA_i\rB_i}\right]\tp \ket{00}^{\rA_{i+1}\rB_{i+1}}\right] &= \ket{\xi_{i+1}}^{\ab}\bigotimes_{j=1}^{i+1}\ket{\psi'}^{\rA_j\rB_j}\\ \label{eq:mnstate2aux2}
    U'_{i+1}\left[U_i\left[{\M}_{\aA_i|x}\tp{\N}_{\bB_i|y}\ket{\psi}^{\ab}\tp\ket{00}^{\rA_1\rB_1\cdots\rA_i\rB_i}\right]\tp\ket{00}^{\rA_{i+1}\rB_{i+1}}\right] &= \ket{\xi_{i+1}}^{\ab}\tp\left({\M'}_{\aA_i|x}\tp{\N'}_{\bB_i|y}\bigotimes_{j=1}^{i}\ket{\psi'}^{\rA_j\rB_j}\right)\tp\ket{\psi'}^{\rA_{i+1}\rB_{i+1}}
\end{align}

Now the aim is to self-test the action of measurements $\M_{\aA_{i+1}|x}$ and $\N_{\bB_{i+1}|y}$. Let us verify how the isometry $U'_{i+1}\circ U_i$ acts on these measurements. By using eqs. \eqref{i-1state}, \eqref{i-1stateandmeas},\eqref{kl} and \eqref{kl2} we obtain:
\begin{align}\label{eq:mnstate2aux3}
    &U'_{i+1}\left[U_i\left[{\M}_{\aA_{i+1}|x}\tp{\N}_{\bB_{i+1}|y}\ket{\psi}^{\ab}\tp\ket{00}^{\rA_1\rB_1\cdots\rA_i\rB_i}\right]\tp\ket{00}^{\rA_{i+1}\rB_{i+1}}\right] = \\ \nonumber &\qquad = U'_{i+1}\left[\left(\K_{\aA_{i+1}|x}\tp\Ll_{\bB_{i+1}|y}\ket{\xi_i}^{\rA\rB}\right)\tp\left({\M'}_{\aA_i|x}\tp{\N'}_{\bB_i|y}\bigotimes_{j=1}^i\ket{\psi'}^{\rA_j\rB_j}\right)\tp\ket{00}^{\rA_{i+1}\rB_{i+1}}\right]
\end{align}
Eq. \eqref{eq:mnstate2aux} implies $U'_{\rA\rA_{i+1}}\left[\K_{a_{i+1}|x}^{(i+1)}\tp\idd^{\rA_{i+1}}\right]{U'}_{\rA\rA_{i+1}}^\dagger = \idd^{\rA}\tp {\M'}_{a_{i+1}|x}^{\rA_{i+1}}$, similarly to the derivation of eq. \eqref{euron}. Following the procedure used to prove eq. \eqref{kl}, this implies $U'_{\rA\rA_{i+1}}\left[\K_{\aA_{i+1}|x}\tp\idd^{\rA_{i+1}}\right]{U'}_{\rA\rA_{i+1}}^\dagger = {\bar{\K}}_{\aA_{i+1}|x}^{\rA}\tp {\M'}_{a_{i+1}|x}^{\rA_{i+1}}$ and $U'_{\rB\rB_{i+1}}\left[\Ll_{\bB_{i+1}|y}\tp\idd^{\rB_{i+1}}\right]{U'}_{\rB\rB_{i+1}}^\dagger = {\bar{\Ll}}_{\bB_{i+1}|y}^{\rB}\tp {\N'}_{b_{i+1}|y}^{\rB_{i+1}}$, where $\bar{K}_{\aA_{i+1}|x}$ and ${\bar{\Ll}}_{\bB_{i+1}|y}$ are positive operators. The eq. \eqref{eq:mnstate2aux2} can be rewritten as
\begin{align}\nonumber
     U'_{i+1}\circ U_i&\left[{\M}_{\aA_i|x}\tp{\N}_{\bB_i|y}\ket{\psi}^{\ab}\tp\ket{00}^{\rA_1\rB_1\cdots\rA_{i+1}\rB_{i+1}}\right] = U'_{i+1}\circ U_i\left[\left(\sum_{a_{i+1}}{\M}_{\aA_ia_{i+1}|x}\tp\sum_{b_{i+1}}{\N}_{\bB_ib_{i+1}|y}\ket{\psi}^{\ab}\right)\tp\ket{00}^{\rA_1\rB_1\cdots\rA_{i+1}\rB_{i+1}}\right] \\ \nonumber
     &= U'_{i+1}\left[\left(\sum_{a_{i+1}}{\K}_{\aA_ia_{i+1}|x}\tp\sum_{b_{i+1}}{\Ll}_{\bB_ib_{i+1}|y}\ket{\xi_i}^{\ab}\right)\tp\left({\M'}_{\aA_i|x}\tp{\N'}_{\bB_i|y}\bigotimes_{j=1}^i\ket{\psi'}^{\rA_j\rB_j}\right)\tp\ket{00}^{\rA_{i+1}\rB_{i+1}}\right] \\ \nonumber
     &= \sum_{a_{i+1},b_{i+1}}\left(\bar{\K}_{\aA_i,a_{i+1}|x}\tp\bar{\Ll}_{\bB_i,b_{i+1}|y}\ket{\xi_{i+1}}^{\ab}\right)\tp\left({\M'}_{\aA_i|x}\tp{\N'}_{\bB_i|y}\bigotimes_{j=1}^i\ket{\psi'}^{\rA_j\rB_j}\right)\tp\left({\M'}_{a_{i+1}|x}\tp{\N'}_{b_{i+1}|y}\ket{\psi'}^{\rA_{i+1}\rB_{i+1}}\right)
\end{align}
The second equality is the consequence of eqs. \eqref{i-1state}, \eqref{kl} and \eqref{kl2}, while the third equality comes from $U'_{\rA\rA_{i+1}}[\K_{{\aA}_{i+1}|x} \otimes 1^{\rA_{i+1}}]{U'}_{\rA\rA_{i+1}}^\dagger = \bar{\K}^\rA_{\aA_{i+1}|x} \otimes {\M'}_{a_{i+1}|x}^{\rA_{i+1}}$ and analogously for $\Ll_{{\bB}_{i+1}|y}$. The final equation is equivalent to \eqref{eq:mnstate2aux2} if and only if $\bar{K}_{\aA_i,a_{i+1}|x} = \bar{\Ll}_{\bB_i,b_{i+1}|y} = \idd$ for all $\aA_{i+1},\bB_{i+1},x,y$. This can be proven by using the same argumentation used after eq. \eqref{brienne} given that  $0 \leq \bar{K}_{\aA_i,a_{i+1}|x}, \bar{\Ll}_{\bB_i,b_{i+1}|y} \leq 1$ and $\sum_{a_{i+1}}{\M'}_{a_{i+1}|x} = \idd, \quad  \sum_{b_{i+1}}{\N'}_{b_{i+1}|y} = \idd$.

Hence, by denoting $U_{i+1} = U'_{i+1}\circ U_i$, we reduce eq. \eqref{eq:mnstate2aux3} to
\begin{align}\label{eq:mnstate2auxFin}
    U_{i+1}\left[{\M}_{\aA_{i+1}|x}\tp{\N}_{\bB_{i+1}|y}\ket{\psi}^{\ab}\tp\ket{00}^{\rA_1\rB_1\cdots\rA_{i+1}\rB_{i+1}}\right] = \ket{\xi_{i+1}}^{\ab}\tp\left(\bigotimes_{j=1}^{i+1}{\M'}_{a_j|x}\tp{\N'}_{b_j|y}\ket{\psi'}^{\rA_j\rB_j}\right)
\end{align}
With this we have proved the inductive step, and with it completed the theorem proof.

\renewcommand{\theequation}{B\arabic{equation}}

\section{Example: parallel self-testing beyond Bell inequalities}\label{appexample}

In this section we give example of lifting the self-testing protocol which is based not on the maximal violation of a Bell inequality but reproduction of the whole set of correlations. For the sake of simplicity we chose the self-testing protocol in the simplest scenario with two parties, each performing two binary measurements ((2,2,2) scenario). The self-testing correlations are
\begin{align}\label{corr0}
    \bra{\psi}\rA_0\tp\rB_0\ket{\psi} = \cos{\gamma}, &\qquad  \bra{\psi}\rA_0\tp\rB_1\ket{\psi} = -\cos{\delta}, \\ \label{corr1}
    \bra{\psi}\rA_1\tp\rB_0\ket{\psi} = \sin{\gamma}, &\qquad  \bra{\psi}\rA_1\tp\rB_1\ket{\psi} = \sin{\delta},
\end{align}
for $\gamma \neq \delta$ and $\gamma,\delta \in (0,\pi/4]$. $\rA_i$ and $\rB_j$ are observables defined as $\rA_i = {\M}_{0|i}-{\M}_{1|i}$ and $\rB_i = {\N}_{0|i}-{\N}_{1|i}$. In \cite{1367-2630-18-2-025021} it is proven that this set of correlations self-tests the maximally entangled pair of qubits. The reference measurement observables are
\begin{align*}
    \rA'_0 = \sz, &\qquad \rA'_1 = \sx,\\
    \rB'_0 = \cos{\gamma}\sz+\sin{\gamma}\sx, &\qquad \rB'_1 = \cos{\delta}\sz-\sin{\delta}\sx.
\end{align*}
We omit the self-testing proof here and direct reader's attention to \cite{1367-2630-18-2-025021}. The isometry used in the proof is the Swap isometry $U$ and physical experiment reproducing correlations \eqref{corr0}-\eqref{corr1} satisfies the following equations
\begin{align}
    U\left[\ket{\psi}^{\ab}\tp\ket{00}^{\rA_1\rB_1}\right] &= \ket{\xi}^{\ab}\tp\ket{\phi^+}^{\rA_1\rB_1}, \\
    U\left[(\M_{a|x}\tp\N_{b|y}\ket{\psi}^{\ab})\tp\ket{00}^{\rA_1\rB_1}\right] &= \ket{\xi}^{\ab}\tp(\M'_{a|x}\tp\N'_{b|y}\ket{\phi^+}^{\rA_1\rB_1}).
\end{align}

Note that one might show that correlations \eqref{corr0}-\eqref{corr1} maximally violate some Bell inequality and the procedure corresponding to the Theorem 1 can be applied to build a self-testing protocol for a tensor product of $n$ maximally entangled qubit pairs. However, we still find it useful to show how to deal with a self-testing protocol based on the reproduction of the whole set of correlations. In $(2,2,2)$ case one can use standard methods to find the Bell inequality maximally violated by some extremal point of the set of quantum correlations (for example NPA hierarchy methods \cite{npa1,npa2}), but in more complicated scenarios this might not be an easy task. Furthermore, most of the known self-testing protocols for multipartite states are not based on the maximal Bell inequality violation. 

In what follows we show how to use the above given self-test to build another self-test, using the same number of measurement choices,  certifying $\ket{\phi^+}\tp\ket{\phi^+}$. Let us introduce the following notation
\begin{align}
\rA_x^{(k)} = \sum_{\aA}(-1)^{a_k}\M_{\aA|x}, &\qquad \rB_y^{(k)} = \sum_{\bB}(-1)^{b_k}\N_{\bB|y}, \qquad \textrm{for}\quad k = 1,2.
\end{align}
The condition for self-testing in this case is reproduction of the following correlations:
\begin{align}\label{corr00}
    &\bra{\psi}\rA_0^{(1)}\tp\rB_0^{(1)}\ket{\psi} = \cos{\gamma}, \qquad  \bra{\psi}\rA^{(1)}_0\tp\rB^{(1)}_1\ket{\psi} = -\cos{\delta}, \\ \label{corr11}
    &\bra{\psi}\rA^{(1)}_1\tp\rB^{(1)}_0\ket{\psi} = \sin{\gamma}, \qquad  \bra{\psi}\rA^{(1)}_1\tp\rB^{(1)}_1\ket{\psi} = \sin{\delta}\\ \label{corr22}
    \frac{1}{p(a_1=a,b_1=b|00)}&\bra{\psi}({\M}^{(1)}_{a|0}\tp{\N}^{(1)}_{b|0})(\rA^{(2)}_0\tp\rB^{(2)}_0)({\M}^{(1)}_{a|0}\tp{\N}^{(1)}_{b|0})\ket{\psi} = \cos{\gamma}, \qquad \forall a,b \\ \label{corr22b} \frac{1}{p(a_1=a,b_1=b|01)}&\bra{\psi}({\M}^{(1)}_{a|0}\tp{\N}^{(1)}_{b|1})(\rA^{(2)}_0\tp\rB^{(2)}_1)({\M}^{(1)}_{a|0}\tp{\N}^{(1)}_{b|1})\ket{\psi} = -\cos{\delta}, \qquad \forall a,b\\ \label{corr33}
    \frac{1}{p(a_1=a,b_1=b|10)}&\bra{\psi}({\M}^{(1)}_{a|1}\tp{\N}^{(1)}_{b|0})(\rA^{(2)}_1\tp\rB^{(2)}_0)({\M}^{(1)}_{a|1}\tp{\N}^{(1)}_{b|0})\ket{\psi} = \sin{\gamma}, \qquad \forall a,b \\ \label{corr33b} \frac{1}{p(a_1 = a,b_1=b|11)}&\bra{\psi}({\M}^{(1)}_{a|1}\tp{\N}^{(1)}_{b|1})(\rA^{(2)}_1\tp\rB^{(2)}_1)({\M}^{(1)}_{a|1}\tp{\N}^{(1)}_{b|1})\ket{\psi} = \sin{\delta}, \qquad \forall a,b .
\end{align}
In terms just of the observed probabilities the set of conditions \eqref{corr00}-\eqref{corr33b} can be written as
\begin{align}
    p(a_1=b_1|00) - p(a_1\neq b_1|00) = \cos\gamma, &\qquad p(a_1=b_1|01) - p(a_1\neq b_1|01) = -\cos\delta,\\
    p(a_1=b_1|10) - p(a_1\neq b_1|10) = \sin\gamma, &\qquad p(a_1=b_1|11) - p(a_1\neq b_1|11) = \cos\delta,\\
(p(a_2=b_2|x=0,y=0,a_1=a,b_1 = b) &- p(a_2\neq b_2|x=0,y=0,a_1=a,b_1 = b)) = \cos\gamma, \qquad \forall a,b\\
   (p(a_2=b_2|x=0,y=1,a_1=a,b_1 = b) &- p(a_2\neq b_2|x=0,y=1,a_1=a,b_1 = b)) = -\cos\delta, \qquad \forall a,b\\
(p(a_2=b_2|x=1,y=0,a_1=a,b_1 = b) &- p(a_2\neq b_2|x=1,y=0,a_1=a,b_1 = b)) = \sin\gamma, \qquad \forall a,b\\
   (p(a_2=b_2|x=1,y=1,a_1=a,b_1 = b) &- p(a_2\neq b_2|x=1,y=1,a_1=a,b_1 = b)) = \sin \delta, \qquad \forall a,b.
\end{align}

The proof goes along the same line as the proof of Theorem 1. Equations \eqref{corr00}-\eqref{corr11} imply the existence of the isometry $U_1 = U_{\rA\rA_1}\tp U_{\rB\rB_1}\tp\idd_{\rP}$ such that
\begin{align}\label{corrstate1}
    U_1\left(\ket{\psi}^{\ab}\tp\ket{00}^{\rA_1\rB_1}\right) &= \ket{\xi_1}^{\ab}\tp\ket{\phi^+}^{\rA_1\rB_1}, \\ \label{corrmnstate1}
    U_1\left(\left(\M^{(1)}_{a|x}\tp\N^{(1)}_{b|y}\ket{\psi}^{\ab}\right)\tp\ket{00}^{\rA_1\rB_1}\right) &= \ket{\xi_1}^{\ab}\tp \left(\M'_{a|x}\tp\N'_{b|y}\ket{\phi^+}^{\rA_1\rB_1}\right).
\end{align}
These two equations imply the following set of equations
\begin{align}
    U_{\rA\rA_1}\left(\M^{(1)}_{a|x}\tp \openone^{\rA_1}\right)U_{\rA\rA_1}^\dagger = \openone^{\rA}\tp{\M'}_{a|x}^{\rA_1}, &\qquad U_{\rB\rB_1}\left(\N^{(1)}_{b|y}\tp \openone^{\rB_1}\right)U_{\rB\rB_1}^\dagger = \openone^{\rB}\tp{\N'}_{b|y}^{\rB_1},\\
    U_{\rA\rA_1}\left(\M_{a_1,a_2|x}\tp \openone^{\rA_1}\right)U_{\rA\rA_1}^\dagger = {\K}_{a_1,a_2|x}^{\rA}\tp{\M'}_{a_1|x}^{\rA_1}, &\qquad U_{\rB\rB_1}\left({\N}_{b_1,b_2|y}\tp \openone^{\rB_1}\right)U_{\rB\rB_1}^\dagger = {\Ll}_{b_1,b_2|y}^{\rB}\tp{\N'}_{b_1|y}^{\rB_1},
\end{align}
where the operators ${\K}_{a_1,a_2|x}$, ${\Ll}_{b_1,b_2|y}$ are positive semidefinite and satisfy
\begin{equation}\label{complet}
    \sum_{a_2}{\K}_{a_1,a_2|x} = \openone , \qquad \sum_{b_2}{\Ll}_{b_1,b_2|x} = \openone
\end{equation}
Given all these equations the first  expression from \eqref{corr22} can be rewritten as
\begin{align}\nonumber
    \frac{1}{p(a_1=a,b_1=b|00)}& \tr\left[(\rA^{(2)}_0\tp\rB^{(2)}_0)({\M}^{(1)}_{a|0}\tp{\N}^{(1)}_{b|0}\proj{\psi}{\M}^{(1)}_{a|0}\tp{\N}^{(1)}_{b|0})\right] = \\ &=  \frac{1}{p(a_1=a,b_1=b|00)}\tr\left[U\left(\rA^{(2)}_0\tp\rB^{(2)}_0\right)U^\dagger\proj{\xi_1}^{\ab}\tp(\M'_{a|0}\tp\N'_{b|0}\proj{\phi^+}^{\rA_1\rB_1}\M'_{a|0}\tp\N'_{b|0})\right]\\ \nonumber
    &= \frac{1}{p_{1}(a_1=a,b_1=b|00)}\sum_{a_2b_2}(-1)^{a_2+b_2}\tr\left[(\K_{a,a_2|0}\tp\Ll_{b,b_2|0})\proj{\xi_1}^{\ab}\tp(\M'_{a|0}\tp\N'_{b|0}\proj{\phi^+}^{\rA_1\rB_1}\M'_{a|0}\tp\N'_{b|0}\right] \\ \nonumber
    &= \frac{1}{p(a_1=a,b_1=b|00)}\sum_{a_2b_2}(-1)^{a_2+b_2}\tr\left[(\K_{a,a_2|0}\tp\Ll_{b,b_2|0})\proj{\xi_1}^{\ab}\right]\tr\left[(\M'_{a|0}\tp\N'_{b|0})\proj{\phi^+}^{\rA_1\rB_1}\right] \\ \label{00}
    &= \sum_{a_2b_2}(-1)^{a_2+b_2}\tr\left[(\K_{a,a_2|0}\tp\Ll_{b,b_2|0})\proj{\xi_1}^{\ab}\right] = \cos{\gamma}.
\end{align}
This equation holds for all $a,b\in{0,1}$. Analogous equations can be obtain starting from the other three relations from \eqref{corr22}-\eqref{corr33}:
\begin{align}\label{01}
    \sum_{a_2b_2}(-1)^{a_2+b_2}\tr\left[(\K_{a,a_2|0}\tp\Ll_{b,b_2|1})\proj{\xi_1}^{\ab}\right] &= -\cos{\delta}\\ \label{10}
    \sum_{a_2b_2}(-1)^{a_2+b_2}\tr\left[(\K_{a,a_2|1}\tp\Ll_{b,b_2|0})\proj{\xi_1}^{\ab}\right] &= \sin{\gamma}\\ \label{11}
    \sum_{a_2b_2}(-1)^{a_2+b_2}\tr\left[(\K_{a,a_2|1}\tp\Ll_{b,b_2|1})\proj{\xi_1}^{\ab}\right] &= \sin{\delta}.
\end{align}
Let us now introduce new operators
\begin{equation}
    2\bar{\rA}_{x} = \sum_{a_1,a_2}(-1)^{a_2}\K_{a_1,a_2|x}, \qquad 2\bar{\rB}_y = \sum_{b_1,b_2}(-1)^{b_2}\Ll_{b_1,b_2|y}
\end{equation}
which are valid measurements observables (cf. \eqref{complet}). Since eqs. \eqref{00}, \eqref{01}, \eqref{10} and \eqref{11} hold for all values of $a$ and $b$ by summing over all the different values we obtain
\begin{align}\label{corrbar00}
    \bra{\xi_1}\bar{\rA}_0\tp\bar{\rB}_0\ket{\xi_1} = \cos{\gamma}, &\qquad  \bra{\xi_1}\bar{\rA}_0\tp\bar{\rB}_1\ket{\xi_1} = -\cos{\delta}, \\ \label{corrbar11}
    \bra{\xi_1}\bar{\rA}_1\tp\bar{\rB}_0\ket{\xi_1} = \sin{\gamma}, &\qquad  \bra{\xi_1}\bar{\rA}_1\tp\bar{\rB}_1\ket{\xi_1} = \sin{\delta}
\end{align}
These relations are exactly self-testing ones and they imply the existence of a local unitary $U_2 = U_{\rA\rA_2}\tp U_{\rB\rB_2}\tp\idd_{\rP}$ such that
\begin{equation}
    U_2\left[\ket{\xi_1}^{\ab}\tp\ket{00}^{\rA_2\rB_2}\right] = \ket{\xi}^{\ab}\tp\ket{\phi^+}^{\rA_2\rB_2}
\end{equation}
Combining this equation with \eqref{corrstate1} we obtain
\begin{equation}
    U_2\circ U_1\left[\ket{\psi}^{\ab}\tp\ket{0000}^{\rA_1\rA_2\rB_1\rB_2}\right] = \ket{\xi}^{\ab}\tp\ket{\phi^+}^{\rA_1\rB_1}\tp\ket{\phi^+}^{\rA_2\rB_2}.
\end{equation}
The proof for self-testing of measurements can be done in the same way as it is done in the proof of Theorem 1.

\renewcommand{\theequation}{C\arabic{equation}}

\section{Combining self-testing protocols to test a tensor product of different quantum states}\label{comb}


Let us introduce the notion of compatible self-testing protocols, as those using the same number of inputs to self-tests the corresponding reference states. Here we outline how one can build a self-testing protocol certifying a tensor product of $n$ different states, which can be independently self-tested by using compatible self-testing protocols.

Namely, the aim is to self-test a state of the form $\bigotimes_{i}\ket{\psi'_i}$, knowing that every $\ket{\psi'_i}$ can be self-tested through observing the maximal violation $\beta_{i}$ of the inequality 
\begin{equation}
    \mathcal{I}_i \equiv \sum_{a_i,b_i = 0}^{o_i-1}\sum_{x,y = 0}^{m-1}b_{i,a_ib_i}^{xy}p_i(a_ib_i|xy) \leq \beta_{i}
\end{equation}
Let us introduce the generalized conditional Bell value:

\begin{align}
 \mathcal{I}_{\aA_i\bB_i}^{i+1} =
 \sum_{a_{i+1}b_{i+1}xy}\frac{b_{i+1,a_{i+1}b_{i+1}}^{xy}}{p(\aA_i\bB_i|xy)}\tr\left[\left({\M}_{\aA_i,a_{i+1}|x}\otimes{\N}_{\bB_i,b_{i+1}|y}\right)\rho_{\rA\rB}\tp\ketbra{00}{00}^{\rA_1\rB_1\cdots\rA_i\rB_i}\right].
\end{align}
The conditions for self-testing $\bigotimes_{i}\ket{\psi'_i}$ by using only $m$ different inputs per party are the following
\begin{itemize}
    \item $\mathcal{I}_1 = \beta_{1}$,
    \item  $ \sum_{\aA_i\bB_i}\mathcal{I}_{\aA_i\bB_i}^{i+1} = \left(\prod_{j = 1}^{i}m_j\right)^{2}\beta_{i+1}$.
\end{itemize}
The self-testing proof goes along the same lines as the proof for Theorem 1. The only difference is that for every $i$ the second condition corresponds to the maximal violation of $\mathcal{I}_i$ by the junk state appearing in the self-testing statement for $i-1$. 

\renewcommand{\theequation}{D\arabic{equation}}

\section{Proof of Theorem \ref{theorem2}}\label{genparallel}

In this section we give a proof of Theorem \ref{theorem2} which deals with combining two or more incompatible self-tests, \ie self-tests corresponding to the scenarios with different inputs and/or output size. The general scenario is as follows: the protocol $\Si{i}$ can be used to self-test the state $\ket{\psi'_i}$ in the scenario where each party has $m_i$ inputs denoted with $x_i,y_i$ and $o_i$ outputs denoted as $a_i,b_i$. If the correlations $p_i(a_ib_i|x_iy_i)$ satisfy the self-testing conditions given by the protocol $\Si{i}$ for each $i$ then there is isometry mapping the physical state to the $\bigotimes_i\ket{\psi'_i}$. This can be seen as a generalization of parallel self-testing, where usually the reference state is an $n$-fold tensor product of some state.

According to the theorem conditions the state $\ket{\psi'_{i}}$, for $i \in \{1,\cdots, n\}$, is self-tested through achieving the maximal violation of the Bell inequality $\mathcal{I}_i$, evaluated as
\begin{equation}
    \mathcal{I}_i \equiv \sum_{a_ib_ix_iy_i}b_{i,a_ib_i}^{x_iy_i}p_i(a_ib_i|x_iy_i) = \beta_i.
\end{equation}
The self-testing scenario is as follows: in every round Alice and Bob receive $n$ classical inputs each, denoted with $\xX = \{x_1,\cdots, x_n\}, \yY = \{y_1,\cdots, y_n\}$, where   $x_i,y_i \in \{0,1,\cdots, m_i-1\}$. They return strings  $\aA = \{a_1,\cdots,a_n\}, \bB= \{b_1,\cdots,b_n\}$, where  $a_i,b_i \in \{0,1,\cdots, o_i-1\}$. The measurement operators are denoted by $\M_{\aA|\xX}$ for Alice and $\N_{\bB|\yY}$ for Bob. Let us introduce the following notation 
\begin{align}
    {\M}_{a_i|x_i} =\frac{1}{ \prod_{j|j\neq i} m_i}\sum_{\aA_{(i)},\xX_{(i)}}\M_{\aA|\xX}, \qquad      {\N}_{b_i|y_i} = \frac{1}{\prod_{j|j\neq i} m_i}\sum_{\bB_{(i)},\yY_{(i)}}\N_{\bB|\yY}\end{align}
These operators are valid measurement operators, as a sum of positive operators they are positive and they satisfy completeness relations $\sum_{a_i}{\M}_{a_i|x_i} = \openone$ and $\sum_{b_j}{\N}_{b_j|y_j} = \openone$, for all $x_i$ and $y_j$. Let us define Bell-like expressions
\begin{equation}\label{svindal}
    \mathcal{I}^i_{\xX_{(i)}\yY_{(i)}} = \sum_{\aA,\bB,x_i,y_i}b_{i,a_ib_i}^{x_iy_i}p_i(\aA\bB|\xX,\yY) 
\end{equation}
The Bell violation for the $i$-th inequality is a normalized sum of values given in \ref{svindal}:
\begin{equation}
 \mathcal{J}^i = \frac{1}{\prod_{j|j\neq i}m_j^2}\sum_{\xX_{(i)},\yY_{(i)}}\mathcal{I}^i_{\xX_{(i)},\yY_{(i)}},  
\end{equation}
Let us set $i=1$. The Bell violation $\mathcal{J}^1$ can be modelled as
\begin{equation}
\mathcal{J}^1 = \sum_{a_1,b_1,x_1,y_1}b_{1,a_1b_1}^{x_1y_1}\bra{\psi}{\M}_{a_1|x_1}\tp{\N}_{b_1|y_1}\ket{\psi}
\end{equation}

If the inequality $\mathcal{J}^1$ is maximally violated  there exist a local unitary $U_{1} = U_{\rA\rA_1}\tp U_{\rB\rB_1}\tp\idd_\rP$ such that
\begin{align}\label{psi1}
    U_{1}\left[\ket{\psi}^{\rA\rB\rP}\tp\ket{00}^{\rA_1\rB_1}\right] &= \ket{\xi_{1}}^{\ab}\tp\ket{\psi'_1}^{\rA_1\rB_1}\\
    U_{1}\left[\left({\M}_{a_1|x_1}\tp{\N}_{b_1|y_1}\tp\openone_\rP\ket{\psi}^{\ab}\right)\tp\ket{00}^{\rA_1\rB_1}\right] &= \ket{\xi_{1}}^{\ab}\tp\left(\M'_{a_1|x_1}\tp\N'_{b_1|y_1}\ket{\psi'_1}^{\rA_1\rB_1}\right).
\end{align}
This implies the following relations 
\begin{align}
    U_{\rA\rA_1}\left[{\M}_{a_1|x_1}\tp \openone^{\rA_1}\right]U_{\rA\rA_1}^\dagger = \openone^{\rA}\tp{\M'}_{a_1|x_1}^{\rA_1}, &\quad U_{\rB\rB_1}\left[{\N}_{b_1|y_1}\tp \openone^{\rB_1}\right]U_{\rB\rB_1}^\dagger = \openone^{\rB}\tp{\N'}_{b_1|y_1}^{\rB_1},\\
    U_{\rA\rA_1}\left[{\M}_{a_1,a_2|x_1,x_2}\tp \openone^{\rA_1}\right]U_{\rA\rA_1}^\dagger = {\K}_{a_1,a_2|x_1,x_2}^{\rA}\tp{\M'}_{a_1|x_1}^{\rA_1}, &\quad U_{\rB\rB_1}\left[{\N}_{b_1,b_2|y_1,y_2}\tp \openone^{\rB_1}\right]U_{\rB\rB_1}^\dagger = {\Ll}_{b_1,b_2|y_1,y_2}^{\rB}\tp{\N'}_{b_1|y_1}^{\rB_1}
\end{align}
where 
\begin{equation}
    {\M}_{a_1,a_2|x_1,x_2} = \frac{1}{m_3m_4\cdots m_n}\sum_{\substack{a_3,a_4,\cdots,a_n\\ x_3,x_4,\cdots, x_n}}\M_{\aA|\xX} \qquad  {\N}_{b_1,b_2|y_1,y_2} = \frac{1}{m_3m_4\cdots m_n}\sum_{\substack{b_3,b_4,\cdots,b_n\\y_3,y_3,\cdots, y_n}}\N_{\bB|\yY}
\end{equation}
and ${\K}_{a_1,a_2|x_1,x_2}$ and ${\Ll}_{b_1,b_2|y_1,y_2}$ are positive operators satisfying 
\begin{equation}
    \sum_{a_2}{\K}_{a_1,a_2|x_1,x_2}=\openone, \quad \sum_{b_2}{\Ll}_{b_1,b_2|y_1,y_2}=\openone, \qquad \textrm{for} \; \textrm{all} \quad x_2,y_2,x_1,y_1,a_1,b_1
\end{equation}
In the second step the maximal violation of the inequality $\mathcal{J}^2$ can be modelled as
\begin{align}
\mathcal{J}^2 &= \frac{1}{m_1^2}\sum_{\substack{a_1,a_2,b_1,b_2\\x_1,x_2, y_1,y_2}}b_{2,a_2b_2}^{x_2y_2}\bra{\psi}{\M}_{a_1,a_2|x_1,x_2}\tp{\N}_{b_1,b_2|y_1,y_2}\ket{\psi}\\
&= \frac{1}{m_1^2}\sum_{\substack{a_1,a_2,b_1,b_2\\x_1,x_2, y_1,y_2}}b_{2,a_2b_2}^{x_2y_2}\bra{\xi_1}{\K}_{a_1,a_2|x_1,x_2}\tp{\Ll}_{b_1,b_2|y_1,y_2}\ket{\xi_1}\bra{\psi'_1}{\M'}_{a_1|x_1}\tp{\N'}_{b_1|y_1}\ket{\psi'_1}\\
&= \frac{1}{m_1^2}\sum_{a_{1},b_{1},x_{1},y_{1}}\bra{\psi'_1}{\M'}_{a_1|x_1}\tp{\N'}_{b_1|y_1}\ket{\psi'_1}\sum_{a_2,b_2,x_2,y_2}b_{2,a_2b_2}^{x_2y_2}\bra{\xi_1}{\K}_{a_1,a_2|x_1,x_2}\tp{\Ll}_{b_1,b_2|y_1,y_2}\ket{\xi_1}\\
&=\sum_{a_{1},b_{1},x_{1},y_{1}}\frac{p_1(a_1,b_1|x_1,y_1)}{m_1^2}\sum_{a_2,b_2,x_2,y_2}b_{2,a_2b_2}^{x_2y_2}\bra{\xi_1}{\K}_{a_1,a_2|x_1,x_2}\tp{\Ll}_{b_1,b_2|y_1,y_2}\ket{\xi_1}
\end{align}
Since numbers $p_1(a_1b_1|x_1y_1)/m_1^2$ are positive and sum to one, and all ${\K}_{a_1,a_2|x_1,x_2}$ and ${\Ll}_{b_1,b_2|y_1,y_2}$ are valid measurement operators the observation $\mathcal{J}^{2} =\beta_{2}$ implies that
\begin{equation}
\sum_{a_2,b_2,x_2,y_2}b_{2,a_2b_2}^{x_2y_2}\bra{\xi_1}{\K}_{a_1,a_2|x_1,x_2}\tp{\Ll}_{b_1,b_2|y_1,y_2}\ket{\xi_1}  = \beta_{2}
\end{equation}
for all $a_1,b_1,x_1,y_1$. Let us define operators
\begin{align}
    \K_{a_2|x_2}^{(2)} = \frac{1}{m_1}\sum_{a_1,x_1}{\K}_{a_1,a_2|x_1,x_2}, \qquad {\Ll}_{b_2|y_2}^{(2)} = \frac{1}{m_1}\sum_{b_1,y_1}{\Ll}_{b_1,b_2|y_1,y_2}
\end{align}
These operators are valid measurement operators, and they satisfy
\begin{equation}
\sum_{a_2,b_2,x_2,y_2}b_{2,a_2b_2}^{x_2y_2}\bra{\xi_1}{\K}_{a_2|x_2}^{(2)}\tp{\Ll}_{b_2|y_2}^{(2)}\ket{\xi_1}  = \beta_{2}.
\end{equation}
The maximal violation $\mathcal{J}^2$ implies that there exists a local unitary $U_2 = U_{\rA\rA_2}\tp U_{\rB\rB_2}\tp\idd_{\rP}$ such that
\begin{align}\label{psi2}
U_2\left[\ket{\xi_1}^{\rA\rB}\tp\ket{00}^{\rA_2\rB_2}\right] &= \ket{\xi_2}^{\rA\rB}\tp \ket{\psi'_2}^{\rA_2\rB_2}\\
U_2\left[({\K}_{a_2|x_2}^{(2)}\tp{\Ll}_{b_2|y_2}^{(2)}\ket{\xi_1}^{\rA\rB})\tp\ket{00}^{\rA_2\rB_2}\right] &= \ket{\xi_2}^{\rA\rB}\tp (\M'_{a_2|x_2}\tp\N'_{b_2|y_2}\ket{\psi'_2}^{\rA_2\rB_2}).
\end{align}
Combining \eqref{psi1} and \eqref{psi2} we get a self-testing statement for a tensor product of two different states
\begin{equation}
    U_2\circ U_1\left[\ket{\psi}^{\ab}\tp\ket{0000}^{\rA_1\rA_2\rB_1\rB_2}\right] = \ket{\xi_2}^{\ab}\tp\ket{\psi'_1}^{\rA_1\rB_1}\tp\ket{\psi'_2}^{\rA_2\rB_2}.
\end{equation}
The process can be further repeated for $i = 3$ to $i=n$, reaching the final statement
\begin{equation}
     U_n\circ \cdots \circ U_1\left[\ket{\psi}^{\ab}\tp\ket{0000}^{\rA_1\cdots \rA_n\rB_1\cdots\rB_n}\right] = \ket{\xi_n}^{\ab}\tp\ket{\psi'_1}^{\rA_1\rB_1}\tp\cdots \tp\ket{\psi'_n}^{\rA_n\rB_n}.
\end{equation}

\end{appendix}

\end{document}